\documentclass[journal,10pt,letterpaper]{IEEEtran}
\usepackage[T1]{fontenc}
\usepackage[utf8]{inputenc}
\usepackage{cite}
\usepackage{amsmath,amssymb,amsthm,booktabs,array}
\usepackage{balance}
\usepackage[hidelinks,pdftitle={Amplified Pruning and Path Selectors for Deterministic Radio Gossiping},pdfauthor={Qin Xin}]{hyperref}
\usepackage{microtype}
\newtheorem{lemma}{Lemma}
\newtheorem{theorem}[lemma]{Theorem}
\newtheorem{corollary}[lemma]{Corollary}
\newtheorem{proposition}[lemma]{Proposition}

\newcommand{\Oh}{\mathcal O}
\newcommand{\Ot}{\widetilde{\mathcal O}}
\newcommand{\indeg}{\deg^{-}}
\newcommand{\NB}{\operatorname{NB}^{-}}
\newcommand{\Select}{\textup{\textsc{Select}}}
\newcommand{\MaxSelect}{\textup{\textsc{MaxSelect}}}
\newcommand{\Drain}{\textup{\textsc{Drain}}}
\newcommand{\Prune}{\textup{\textsc{Prune}}}
\newcommand{\Clear}{\textup{\textsc{Clear}}}
\newcommand{\CPrune}{\textup{\textsc{ClearPrune}}}
\newcommand{\Gossip}{\textup{\textsc{Gossip}}}
\newcommand{\dbar}{\overline d}
\title{Amplified Pruning and Path Selectors\\for Deterministic Radio Gossiping}
\author{Qin Xin\thanks{Qin Xin is with the Faculty of Science and Technology, University of the Faroe Islands, Vestara Bryggja 15, FO-100 T\'orshavn, Faroe Islands (e-mail: \href{mailto:qinx@setur.fo}{qinx@setur.fo}).}}
\IEEEspecialpapernotice{\normalfont\normalsize Revised October 27, 2026}
\begin{document}
\pagestyle{plain}
\maketitle
\thispagestyle{plain}
\begin{abstract}
We study deterministic gossiping in unknown strongly connected directed radio networks with $n$ vertices, $m$ arcs, and distinct identifiers in $[N]$, where $N=n^{\mathcal O(1)}$. Messages are unrestricted and collision detection is unavailable. Our neighborhood-clearing procedure reuses an elected receiver to secure its entire closed in-neighborhood in $\mathcal O(n\log^2n\log\log n)$ rounds. Charging each distinct closed neighborhood once yields $\mathcal O(m\log^2n\log\log n)$ rounds under an explicit degree-concentration condition, for every polynomial identifier range. A complementary adaptive-pruning procedure combines quadratic contraction with completion sensitive to the number of unsecured sources. If $a\ge1$ sources remain at completion entry, permutation selectors reduce the balanced completion budget by a factor of $\log^{2/3}(a+1)$ relative to interleaved path selectors. A geometric execution attains $\widetilde{\mathcal O}(k^2+n\tau+\min\{u^2,(nu)^{2/3}\})$ for a path cover of scale $k$, weight $\tau$, and $u$ exceptions. For directed graphs with common in- and out-degree $\Delta$, the portfolio achieves $\widetilde{\mathcal O}(\min\{n\sqrt\Delta,n^2/\Delta\})$. Every graph also admits $\widetilde{\mathcal O}(\min\{n^{4/3},N+(mn)^{3/5}\})$ rounds. Only $n,N$ are inputs; degrees, covers, and structural constants are used solely in the analysis. The unrestricted polynomial-identifier density extension remains open.
\end{abstract}
\begin{IEEEkeywords}
Deterministic gossiping, directed radio networks, neighborhood clearing, path selectors, selective families.
\end{IEEEkeywords}

\section{Introduction}
Radio-network gossiping disseminates every vertex's original rumor to every vertex. The network has vertex set $V$, $n=|V|$ vertices, $m$ arcs, and identifiers in $[N]=\{1,\ldots,N\}$. Reception succeeds when exactly one in-neighbor transmits and the receiver listens; collisions and silence are indistinguishable. The topology is unknown.

We use $\Ot$ to suppress polylogarithmic factors in $n$. G\k{a}sieniec, Radzik, and Xin~\cite{grx} established the general $\Ot(n^{4/3})$ bound for polynomial identifiers; Kuschner et al.~\cite{kuschner} subsequently refined the selector-based approach. For linear identifiers, Wu and Chrobak~\cite{wu} obtain $\Ot((mn)^{3/5})$. Their questions about polynomial identifiers and average degree motivate our regularity, concentration, and density results. Table~\ref{tab:bounds} compares the bounds and assumptions; Section~\ref{sec:openquestions} states which questions they resolve.

\subsection{Why the identifier range matters}
Round robin assigns one slot to each possible identifier, taking $N$ slots even when only $n$ identifiers are present. A polynomial range can therefore make this step dominant. Identifiers reveal neither adjacency nor rank among present vertices; assigning consistent ranks would itself require communication.

A replacement for round robin must preserve the guarantee each argument uses. Reception counts certify interference; condensing also requires packet identities or a prescribed packet crossing a frontier. We establish the appropriate guarantees and keep the exact-delivery cost explicit. Control runs on the full physical network, so reducing the data sender set requires no connectivity assumption on the remaining data graph.

Selective schedules use the actual identifiers in $[N]$: capacity bounds the number of contenders, independently of gaps between their labels. The selective components pay logarithmically for the identifier universe; the exact-frontier term remains explicit in the condensing bound. Permutation selectors add an ordering guarantee~\cite{grx,kuschner}, allowing successful transmissions to carry knowledge along a path during one invocation. Our analysis accounts for interference through the path's closed carrier neighborhood. Common finite schedules require no topology discovery; local preprocessing time is unrestricted.

\subsection{Structural condition and results}
Write $L=\log n$ and $\ell=\log\log n$, and let $\indeg(v)$ denote the original in-degree. The in-degree tail and the average degree are
\begin{equation}\label{eq:tail}
 h_k=|\{v\in V:\indeg(v)\ge k\}|,
 \qquad \dbar=m/n.
\end{equation}
The degree-concentration specialization assumes
\begin{equation}\label{eq:concentration}
 h_{\lceil C_0\dbar\rceil}\le C_1\dbar,
 \qquad C_0,C_1\ge1\text{ fixed}.
\end{equation}
The condition bounds the number of receivers with in-degree at least a constant multiple of the average by a constant times that average. Individual degrees may reach $n-1$. An edge-count bound alone need not imply this concentration.

Neighborhood clearing elects a receiver once and reuses it as a known broadcast source. Each probe is followed by a broadcast of the receiver's saved inventory and an update of the data sender set. Repetition secures the entire closed in-neighborhood in $\Oh(nL^2\ell)$ rounds, so each distinct closed neighborhood can justify at most one qualifying phase. Theorem~\ref{thm:clearadaptive} attains the resulting profile bound without a threshold input, and Corollary~\ref{cor:edge} gives $\Oh(mL^2\ell)$ under~\eqref{eq:concentration}.

A complementary pruning procedure increases probe accuracy as its contender cap decreases. Two qualifying selections of the same receiver or closed-neighborhood class give a quadratic contraction, limiting that receiver or class to $\Oh(\ell)$ qualifications. Lemma~\ref{lem:charge} establishes this charge and supports the survivor and reception-cover results. These results use the ledger produced by that procedure; the clearing bound uses its own execution.

\begin{table*}[t]
\centering\small
\renewcommand{\arraystretch}{1.08}
\caption{Deterministic gossiping bounds. New protocols need only $n,N$. Set $L=\log n$, $\ell=\log\log n$; $\widehat G_D$ is defined in~\eqref{eq:clearprofile}. Condition $\mathsf S$ means $m=\Theta(n)$ with some $D=\Ot(n^{3/5})$ satisfying $h_D=\Ot(n^{1/5})$.}\label{tab:bounds}
\begin{tabular*}{\textwidth}{@{\extracolsep{\fill}}llll@{}}
\toprule
Result & Assumption or extra input & Identifier range & Communication rounds\\
\midrule
G\k{a}sieniec--Radzik--Xin~\cite[Thm.~2]{grx} & None & $N=n^{\Oh(1)}$ & $\Oh(n^{4/3}L^{10/3})$\\
Kuschner et al.~\cite[Sec.~3]{kuschner}$^{\dagger}$ & None & $N=n^{\Oh(1)}$ & $\Oh(n^{4/3}L^2\ell^{2/3})$\\
Wu--Chrobak~\cite{wu} & Edge-count upper bound & $N=\Oh(n)$ & $\Ot((mn)^{3/5})$\\
Wu--Chrobak~\cite{wu}, regular & Directed $\Delta$-regular & $N=\Oh(n)$ & $\Ot(n\sqrt\Delta)$\\
\midrule
Corollary~\ref{cor:uniformspan}, general & None & $N=n^{\Oh(1)}$ & $\Ot(\min\{n^{4/3},N+(mn)^{3/5}\})$\\
Corollary~\ref{cor:uniformspan}, uniform & None & $N=\Ot(n^{6/5})$ & $\Ot(\min\{n^{4/3},(mn)^{3/5}\})$\\
Corollary~\ref{cor:edge} & Condition~\eqref{eq:concentration} & $N=n^{\Oh(1)}$ & $\Oh(mL^2\ell)$\\
Theorem~\ref{thm:profile}, earlier bound & Condition~\eqref{eq:concentration} & $N=n^{\Oh(1)}$ & $\Oh(mL^2\ell^2)$\\
Theorem~\ref{thm:clearadaptive} & Closed-neighborhood profile $t_D$ & $N=n^{\Oh(1)}$ & $\Oh(\min_D\widehat G_D)$\\
Theorem~\ref{thm:density}, sparse & $\mathsf S$ & $N=n^{\Oh(1)}$ & $\Ot(n^{6/5})$\\
Corollary~\ref{cor:polytail}, $p=5$ & $h_k=\Oh(n/k^5)$ for all $k$ & $N=n^{\Oh(1)}$ & $\Oh(n^{7/6}L^{19/12}\ell^{7/12})$\\
Theorem~\ref{thm:pathadaptive} & Survivor profile $a_D$ & $N=n^{\Oh(1)}$ & $\Ot(\min_D[n(1+\rho_D)+W_D])$\\
Theorem~\ref{thm:pathcover} & Path cover $(k,\tau,u)$ & $N=n^{\Oh(1)}$ & $\Ot(k^2+n\tau+\min\{u^2,(nu)^{2/3}\})$\\
Theorem~\ref{thm:regular} & Directed $\Delta$-regular & $N=n^{\Oh(1)}$ & $\Ot(\min\{n\sqrt\Delta,n^2/\Delta\})$\\
Corollary~\ref{cor:cover} & Receiver misses $\Oh(\sqrt{nL}\,\ell)$ others & $N=n^{\Oh(1)}$ & $\Oh(nL^2\ell^2)$\\
Corollary~\ref{cor:parameter}, dense & Complete core; $m=\Theta(n^{7/5})$ & $N=n^{\Oh(1)}$ & $\Oh(nL^2\ell)$, lower bound $\Omega(n)$\\
\bottomrule
\end{tabular*}
\par\smallskip
\begin{minipage}{\textwidth}\footnotesize
$^{\dagger}$The accessible preprint of~\cite{kuschner} prints $\Oh(n^{4/3}L\ell^{2/3})$ in Theorem~2, but its Section~3 derives $\Oh(n^{4/3}L^2\ell^{2/3})$. We report the latter, with all logarithmic factors explicit. The bound from~\cite{grx} is its Theorem~2 bound, sharper than the abstract's $\Oh(n^{4/3}L^4)$.
\end{minipage}
\end{table*}

Completion also exploits the $a$ sources outside the secured ledger. Theorem~\ref{thm:pathadaptive} adapts to this population without a threshold search; pruning and collection share disjoint source charges. Permutation selectors~\cite{kuschner} reduce the balanced completion budget by $\log^{2/3}(a+1)$ relative to interleaved path selectors.

The stopping rule compares elapsed pruning time with every reference-prefix completion budget. A reception cover of weight $\tau$ and $u$ exceptions after $j$ broadcasts gives $\Ot(n(j+\tau)+(nu)^{2/3})$ rounds; successive covers compose on the same execution. A receiver missing $u$ others gives $\Oh(nL^2\ell^2+(u+1)^2L)$. Clearing gives $\Oh(nL^2\ell)$ on distributed-cover and complete-core families, including $m=\Theta(n^{7/5})$, whose distance lower bounds match the linear polynomial exponent.

For an original path cover of scale $k$, weight $\tau$, and $u$ exceptions, one geometric execution attains $\Ot(k^2+n\tau+\min\{u^2,(nu)^{2/3}\})$ (Theorem~\ref{thm:pathcover}). Fixed residual carriers permit composition with pruning; the minimum retains direct completion when few exceptions remain. Condensing separately gives $\Ot(\min\{n^{4/3},N+(mn)^{3/5}\})$ on every graph. Covers and graph parameters remain analytical witnesses, not inputs.

\subsection{Open questions addressed by our approach}\label{sec:openquestions}
Wu and Chrobak~\cite[Sec.~7]{wu} ask whether polynomial identifiers can preserve a bound below $\Ot(n^{4/3})$, and whether the average in-degree can support nearly edge-linear gossiping. Our results answer these questions affirmatively in the following precisely stated settings.

\emph{Polynomial identifiers.} Theorem~\ref{thm:regular} extends~\cite[Theorem~2]{wu} to every $N=n^{\Oh(1)}$, giving $\Ot(\min\{n\sqrt\Delta,n^2/\Delta\})$ on directed $\Delta$-regular graphs without knowing $\Delta$. The second term improves the earlier budget and general exponent when $\Delta=n^{\gamma+o(1)}$, $\gamma>2/3$. For arbitrary graphs, Corollary~\ref{cor:uniformspan} preserves the density benchmark throughout $N=\Ot(n^{6/5})$; Theorem~\ref{thm:density} supplies sufficient residual conditions for larger polynomial ranges.

\emph{The average-degree target.} Under~\eqref{eq:concentration}, Corollary~\ref{cor:edge} achieves the $\Ot(n\dbar)=\Ot(m)$ target of~\cite[Sec.~7]{wu}, with explicit bound $\Oh(mL^2\ell)$ for polynomial identifiers. Thus $m=\Theta(n^{1+\varepsilon})$, $0\le\varepsilon<1/3$, gives $\Ot(n^{1+\varepsilon})$, below the general exponent. Section~\ref{sec:structural} realizes these scales with linear diameter and maximum in-degree. Reception covers give nearly linear communication on further families outside concentration.

The regular-graph identifier extension and the degree-concentrated average-degree target are thus obtained by protocols using only $n,N$. The unrestricted density extension for every polynomial identifier range, and nearly edge-linear gossiping on all sparse graphs, remain open.

\subsection{Relation to previous techniques}
Path selection and neighborhood doubling originate in~\cite{grx}; shorter permutation selectors are from~\cite{kuschner}. Quasi-dissemination and replay follow~\cite{cgr,grx}, selective schedules use~\cite{selectors}, and broadcast control uses~\cite{blind}. Condensing and regular expansion are from~\cite{wu}. Our refinements are complete neighborhood clearing through reuse of an elected broadcast source, adaptive accuracy, the two-step contraction, shared closed-neighborhood charges, and their composition with surviving-source completion and cover-based gains.

\subsection{Organization and proof dependencies}
We first develop reception certificates, adaptive pruning, bounded-degree completion, and replay. Section~\ref{sec:clearing} uses these primitives to prove the improved clearing profile. Sections~\ref{sec:path}--\ref{sec:pathcover} develop survivor-sensitive completion and path covers. Section~\ref{sec:cover} introduces reception covers, with the nonregular quantile construction in Section~\ref{sec:nonregularquantile} and exact two-parameter accounting in Section~\ref{sec:familyaccounting}. The independent condensing analysis in Section~\ref{sec:transfer} yields density transfer. Section~\ref{sec:structural} develops structural consequences and compares the earlier concentration refinements in Section~\ref{sec:legacyconcentration}. Appendix~\ref{app:protocol} records implementation and budget details. The dependencies are acyclic: fixed-threshold completion uses established primitives and never invokes an adaptive controller.

The portfolio interleaves complete protocols with private inventories, ledgers, clocks, and recordings, accepting the first certified completion. Each charge and residual witness belongs to the execution that produced it. These separate states make the minimum of the complete running-time bounds executable.

Table~\ref{tab:state} summarizes the common state. Degrees, covers, and reference prefixes remain analytical witnesses.

\section{Model and Control}
Let $G=(V,E)$ be a simple, loop-free, strongly connected directed graph, with $n\ge2$ and $m$ arcs. Distinct identifiers lie in $[N]=\{1,\ldots,N\}$, where $N=n^{\Oh(1)}$. Common $n,N$ are known; topology, degrees, present identifiers, and $m$ are unknown. Each round a vertex transmits or listens. A listener receives a packet exactly when one in-neighbor transmits, and silence otherwise. Messages and local computation are unrestricted. Write $\NB(v)=\{u:(u,v)\in E\}$, $\indeg(v)=|\NB(v)|$, and $\Delta=\max_v\indeg(v)$. Put $L=\log n$, $\ell=\log\log n$, and $\Lambda=\lceil\log N\rceil=\Theta(L)$. Logarithms are base two; asymptotic formulas use $n\ge4$. Smaller networks use $n-1$ strong-selector periods.

Packets identify their immediate sender; rumors carry source identifiers. Let $K_v\subseteq V$ index the genuine rumors known at $v$, initially $K_v=\{v\}$. Every transmitted packet appends all current rumor values and source tags. A common secured set $F$ contains sources disseminated by designated completed broadcasts. Hence
\begin{equation}\label{eq:inv}
 \begin{gathered}
 F\subseteq K_v\quad\text{for every }v,\qquad A=V\setminus F,\\
 d_A(v)=|\NB(v)\cap A|.
 \end{gathered}
\end{equation}
Write $a=|A|=n-|F|$ for the active source count. The closed active contender set at $v$ is $A\cap(\NB(v)\cup\{v\})$. In profile data slots, only $A$ transmits; dormant vertices still listen, relay control, and remain eligible inventory holders. Knowledge grows on reception, but $F$ changes only at common successful-broadcast boundaries. Interleaved simulations keep private inventories and control state.

\subsection{Knowledge, ledgers, and namespaces}
$K_v$ contains genuine rumor values with source tags; header identifiers alone do not count. Each interleaved simulation keeps a separate monotone inventory and may ignore packets received outside the simulation's assigned slots.

\begin{table}[t]
\centering\small
\caption{State and budget notation. All ledgers and caps below are common within their own simulation.}\label{tab:state}
\renewcommand{\arraystretch}{1.10}
\begin{tabular}{@{}>{\raggedright\arraybackslash}p{0.22\columnwidth}>{\raggedright\arraybackslash}p{0.73\columnwidth}@{}}
\toprule
Symbol & Meaning and role\\
\midrule
$K_v$, $H_v$ & Current genuine inventory; frozen inventory used as a witness in one selection.\\
$F$, $A$, $a$ & Secured-source ledger; the ledger complement; the size of that complement.\\
$e$, $\kappa$ & Pruning cap with $d_A(v)<2e$; closed-contender bound for a probe.\\
$\rho_D$, $\chi_D$ & Charges for reception-score pruning and complete neighborhood clearing, respectively.\\
$Q$, $q(\kappa)$ & Baseline amplification; adaptive amplification at the current contender bound.\\
$R$, $B$ & One frozen Boolean-test budget; a successful selection and broadcast budget.\\
$C$ & Fixed data carriers throughout a path-cover execution, even as $F$ grows.\\
$F_0$, $\mathcal L$ & Frozen original-source ledger; changing packet-owner ledger in a condensing phase.\\
$M$, $H_\kappa$ & Upper bound on active arcs, padded to $n$; mixed scale $n+\sqrt{M\kappa}$.\\
\bottomrule
\end{tabular}
\end{table}

The ledger $F$ records only completed designated broadcasts and can therefore omit rumors that are already globally known. Ledger membership is common, even though actual global dissemination is not locally observable. The invariant $F\subseteq K_v$ suffices for pruning and completion.

For example, suppose a winner saves the inventory $\{v,x,y\}$ and receives a fourth rumor $z$ during selection control. Actual packets may carry all four values, but the completed designated broadcast credits only the saved three-source inventory. Every vertex therefore performs the same ledger update. Overestimating the remaining work is harmless for an upper bound. The example also separates two kinds of progress: hearing one sender increases a reception count by one, whereas disseminating the saved inventory may secure three original sources.

Pruning reduces data contention while retaining every physical relay. On the cycle $x\to y\to z\to x$, let $F=\{y\}$, so $A=\{x,z\}$. Only the arc from $z$ to $x$ remains in the active induced graph. Nevertheless, a data transmission from $x$ can reach dormant $y$, making every received rumor quasi-disseminated because $y$'s original rumor is globally known. Control may use all three vertices. Paths to dormant holders therefore remain useful even when the active induced graph is disconnected.

In a condensing phase, packet objects are additionally indexed by their phase and owner identifier. Packet $(s,u)$ contains the frozen original-rumor inventory of $u$ at phase entry. The packet selector works on owner identifiers within one fixed phase, so the identifier universe remains $[N]$. Earlier phase objects can be discarded after their values are incorporated into $K_v$. A phase ledger records globally broadcast packet objects and is distinct from the original-source ledger. In particular, broadcasting a snapshot of $u$ disseminates $u$'s own original rumor, even when the snapshot is later used only as a packet witness.

A packet is called quasi-disseminated if the packet is globally known or is held by a vertex whose own original rumor is globally known. Every original rumor inside such a packet is then quasi-disseminated. Our frozen-active-set proof uses this definition. A snapshot broadcast need only disseminate the values present when the snapshot was created; information subsequently received by the owner is handled through retention and exact replay.

\subsection{Selective schedules}
A schedule isolates $x$ from $X\subseteq[N]$ if one transmitting set intersects $X$ exactly in $\{x\}$. For $1\le k\le N$ and $0\le r<k$, schedules of length
\begin{equation}\label{eq:selector}
 \Oh\!\left(\frac{k^2\log N}{r+1}\right)
\end{equation}
leave at most $r$ members of every $|X|\le k$ unisolated~\cite{selectors}; $r=0$ gives a strong selector. Choose each identifier independently with probability $1/k$ per slot, for $k\ge2$. For any $k$-set $X$ and $(r+1)$-subset $Y\subseteq X$, one slot isolates a member of $Y$ with probability at least $(r+1)/(ek)$. A union bound over at most $N^k2^k$ pairs gives~\eqref{eq:selector}. Pad smaller sets to $k$; for $k=1$, one slot suffices. All needed finite schedules are fixed commonly by exhaustive search; polynomial local construction time is not asserted.

\subsubsection{Uniformity over small contender sets}
The padding argument for a partial selector must preserve the property actually required. Extend a contender set $X$ of size at most $k$ to a $k$-set. If an element of $X$ is isolated in the extended set, the same element is still isolated in $X$. Since at most $r$ elements of the extension are missed, at most $r$ elements of $X$ are missed. This implication remains valid when $|X|\le r$, because the conclusion then imposes no positive number of successes.

A weak selector has a different requirement: every nonempty contender set must have a successful sender. Padding an arbitrary weak selector is insufficient, because the weak selector's only success might belong to a padding identifier. We therefore concatenate partial selectors at dyadic scales. If $s$ contenders are present, the first covering scale $t$ has $t<2s$; leaving at most $\lfloor t/2\rfloor$ unisolated guarantees a genuine success. This distinction is used in the packet certificate.

All schedules are finite and common. We choose the lexicographically first family passing the finite isolation tests. The probabilistic argument establishes existence and communication length; our protocol uses no runtime randomness. A claim of polynomial local preprocessing would require an additional explicit-construction theorem and is not made here. Thus schedule construction depends only on public parameters. The control procedures below coordinate selections from changing local eligibility flags without discovering the topology or the present identifier set.

\subsection{Common tests and selection}
The generalized broadcast in~\cite[Theorem 23]{blind} permits an arbitrary nonempty initial source set and has running time $\Oh(n\log N\log\log n)$. Fix a common padded budget
\[
 R=\Theta(n\log N\log\log n).
\]
For a Boolean test, freeze flags at entry and reset private control state and clocks. Exactly the true-flag vertices initiate a common marker; others relay the marker only after reception. After $R$ rounds every vertex knows the OR: no initiators produce silence, whereas any initiator delivers the marker everywhere. Rumor knowledge cannot initiate a marker.

For a frozen eligibility predicate $P(v)$, binary search over eligible keys in $[M]$, $M\le(n+1)N$, selects the minimum or certifies emptiness. An initial OR, at most $\lceil\log_2M\rceil$ further tests, and a broadcast of the winner's frozen inventory $S_v$ cost
\begin{equation}\label{eq:B}
 B=R(\lceil\log_2((n+1)N)\rceil+2)
   =\Theta(n\log^2n\log\log n).
\end{equation}
Call this $\Select(P)$. Flags, keys, and inventories remain fixed throughout. A success updates the common set $F\leftarrow F\cup S_v$; an empty selection leaves $F$ unchanged.

For frozen integer scores $0\le b_v\le n$, $\MaxSelect(b_v)$ first finds the minimum frozen key $(n-b_v)N+\operatorname{id}(v)$ and announces the maximum score $b=\max_vb_v$. Our selection procedure then broadcasts the winner's saved inventory. The key search and the inventory broadcast are separate stages, together costing $\Oh(B)$. A thresholded call may stop after the search. The search always considers all vertices, independently of the threshold. This makes thresholded pruning a prefix of one common greedy execution. Keys are local; radio identifiers remain in $[N]$.

Frozen OR tests supply the common decisions unavailable from silent data slots. Searching key prefixes identifies a winner without knowing the present identifiers. Frozen scores and payloads fix that choice during control, and the final designated broadcast makes the saved snapshot and ledger update common. Later arguments use exactly these synchronization guarantees.

Completion can also be checked without learning the topology or the list of present identifiers. Each vertex freezes the flag $|K_v|<n$ and the protocol computes the flags' OR. A false answer certifies that every vertex holds all original values: inventories contain only genuine, distinctly tagged sources, and exactly $n$ sources exist. This test is used after complete replay when verifying parameter guesses. A true answer merely rejects that trial. Neither header identifiers nor a guessed edge budget can manufacture a false success, because only received genuine rumor values contribute to the inventory count.

\begin{lemma}[Frozen-inventory collection]\label{lem:drain}
Freeze $H_v=K_v$ at entry. For every integer $q\ge1$, a procedure $\Drain(q)$ terminates at a common boundary with $|H_v\setminus F|<q$ for all $v$, using at most $R+B\lfloor a/q\rfloor$ rounds if $a=n-|F|$ at entry.
\end{lemma}
\begin{proof}
Repeat $\Select(|H_v\setminus F|\ge q)$ with fresh flags and fixed $H_v$. The winner's actual entry snapshot contains $H_v$, adding at least $q$ new sources to $F$; hence at most $\lfloor a/q\rfloor$ calls succeed. Stop on an empty call. Since $F,H_v$ stay fixed during that call, the saved inequalities remain false at the end. Each success costs at most $B$; the final empty call ends after the call's initial $R$-round OR test. The outcome and ending boundary are common.
\end{proof}

Reception scores count immediate senders; inventory gains count newly secured sources. A packet may contain many rumors, but their origins do not become additional senders. Pruning uses the former to certify interference, and collection uses the latter to pay for broadcasts. Frozen scores and snapshots preserve both guarantees during control.

\section{Reception-Certified Interference Reduction}\label{sec:reception}
Pruning estimates contention and disseminates a useful inventory. Lemma~\ref{lem:probe} certifies the number of distinct active in-neighbors heard. For condensing, Lemma~\ref{lem:union} either delivers the neighboring entry-packet union or supplies enough packets for collection. Both schedules freeze the active set and decision state during an invocation; payloads may append further genuine knowledge.

A completed inventory broadcast turns reception into a reduction of interference: the credited original sources leave the data sender set. The global analysis therefore charges broadcasts through the sources they secure and through repeated qualifications of receivers or neighborhood classes.

\begin{lemma}[Amplified reception certificate]\label{lem:probe}
For a frozen active set and a common upper bound $\kappa\ge1$ on every closed active contender set, any integer $Q\ge2$ admits a schedule of length $\Oh(\kappa Q\log N)$ giving distinct-sender counts
\begin{equation}\label{eq:probe}
 (1-1/Q)d_A(v)\le b_v\le d_A(v).
\end{equation}
\end{lemma}
\begin{proof}
\emph{Schedule length.} Cap $\kappa$ at $N$; every contender set lies in $[N]$. Concatenate dyadic scales $k=1,2,4,\ldots$ through the first scale at least $\kappa$, capping the final scale at $N$ if needed. At scale $k$, use~\eqref{eq:selector} with $r=\lfloor k/(4Q)\rfloor$. The length is $\Oh(Qk\log N)$, also for $k<4Q$, when $r=0$. Geometric summation gives $\Oh(\kappa Q\log N)$.

\emph{Reception guarantee.} Fix a receiver with $d=d_A(v)>0$ and closed contender set $X=A\cap(\NB(v)\cup\{v\})$, $s=|X|$. At the first scale $k\ge s$, $k<2s$ and $s\le d+1\le2d$. Hence $r<s/(2Q)\le d/Q$. At least $d-r$ active in-neighbors are isolated from $X$. Including an active receiver's identifier in $X$ ensures that $v$ listens in each isolation slot; dormant receivers always listen. Each isolated neighbor is therefore heard, giving the lower bound. The count cannot exceed the active in-degree: repeated packets and aggregated origins add no sender identities. The case $d=0$ is immediate.
\end{proof}

Pruning needs neighbor counts; condensing also needs packet identities. The next certificate delivers the neighboring entry union or gives the receiver enough distinct packets for collection. Later frame control rules out the threshold alternative precisely where every pending packet must be delivered.

\begin{lemma}[Amplified packet-union certificate]\label{lem:union}
Freeze an active set and entry packet sets $P_u\subseteq[N]$. For integers $2\le K\le n$, $1\le r\le n$, $Q\ge2$, a schedule of length $\Oh(K[Q\log N+r\log^{3/2}N])$ gives every receiver with $d_A(v)+1\le K$ the certificate~\eqref{eq:probe}. If $W_v$ is the receiver's final knowledge of these designated packets and $U_v=P_v\cup\bigcup_{u\in\NB(v)\cap A}P_u$, then
\begin{equation}\label{eq:union}
 U_v\subseteq W_v\quad\text{or}\quad |W_v|\ge r.
\end{equation}
\end{lemma}
\begin{proof}
\emph{Schedule and cost.} Put $\eta=\max\{Q,r\sqrt{\log N}\}$. Use dyadic contender scales through the first covering $K$, capped at $N$. At $k<4\eta$, run a strong sender selector. Otherwise put $q=\lfloor k/(4\eta)\rfloor$ and first isolate all but $q-1$ contenders using~\eqref{eq:selector}. Then run a strong selector on packet identifiers for sets of size at most $r$. For each packet-selector row $R$, only senders with $P_u\cap R\ne\varnothing$ execute a weak sender selector for nonempty sets of size at most $q$. For a nonempty set of size $s\le q$, the set's first covering dyadic scale $t<2s$ leaves at most $\lfloor t/2\rfloor<s$ members unisolated. Summing these partial selectors gives a weak schedule of length $\Oh(q\log N)$ valid for \emph{every} size at most $q$. Payloads contain all entry packets. At a large scale the cost is
\[
 \Oh\!\left(k\eta\log N+\frac{kr^2}{\eta}\log^2N\right)
 =\Oh(k[Q\log N+r\log^{3/2}N]).
\]
Small strong-selector scales cost $\Oh(k\eta\log N)$. Geometric summation proves the bound. The choice of $\eta$ balances partial isolation against packet-gated selection.

\emph{Reception counts.} Zero active in-degree makes both claims immediate, since the receiver already knows $P_v$. Otherwise, at the first scale covering the closed contender set, fewer than $d_A(v)/Q$ neighbors remain unheard after the first part, as in Lemma~\ref{lem:probe}.

\emph{Packet union.} A strong-selector scale delivers all entry packets directly. Otherwise, suppose the \emph{final} set $W_v$ has size below $r$ and misses $\pi\in U_v$. A packet-selector row isolates $\pi$ from $W_v\cup\{\pi\}$, whose size is at most $r$. Every sender already heard has $P_u\subseteq W_v$ and is inactive in that row. Since $P_v\subseteq W_v$, the receiver also listens. The missing packet must occur at an unheard neighbor, so the gated neighbor set is nonempty and has at most $q-1$ senders. Weak selection isolates one. The sender's gated entry set contains a packet outside $W_v$, contradicting final knowledge. This proves~\eqref{eq:union}, including half-duplex reception.
\end{proof}
\subsection{Why final knowledge is the right separating set}
The unknown neighboring packet union can exceed the threshold, so separating a desired packet from that entire union has no short-length guarantee. Lemma~\ref{lem:union} instead separates the desired packet from final knowledge $W_v$. If $W_v$ is large, the threshold alternative holds; otherwise $W_v\cup\{\pi\}$ fits the packet selector's capacity.

Every already-heard sender has an entry set contained in $W_v$ and is silent in the separating row. This removes interference from all represented entry sets. The remaining sender set is small because partial isolation heard almost every contender. Weak selection then succeeds without knowledge of that residual set.

The receiver's own entry set also lies in $W_v$. Thus, if the receiver belongs to the active sender set, the receiver is inactive in the separating row. This is the precise point where the composed packet construction satisfies the half-duplex model. Counting neighbors without controlling the receiver's own transmission would not establish the claimed delivery.

The separating row is chosen only in the analysis. Nodes execute every predetermined row, without predicting missing packets or observing collision feedback. The schedule is oblivious once the schedule's active set and entry packets are frozen.

For example, take threshold three and suppose final knowledge contains two designated packets but misses a neighboring one. A packet-selector row separates the missing packet from those two. Already-heard senders are silent in that row, and weak selection isolates a remaining sender. Any newly delivered packet outside the supposed final set gives the contradiction; the successful sender need not carry the particular packet used to choose the row.

The three-selector architecture is from~\cite[Lemma 4]{wu}. The final-set argument removes the small-union premise: when a desired packet is missing, delivery of any packet outside the supposed final set contradicts that set's definition. Balancing partial isolation and packet-gated selection gives the stated logarithmic cost.

\subsection{Adaptive probing and pruning}
Fix the baseline amplification
\begin{equation}\label{eq:Q}
 Q=\max\{4,\lceil\log n\rceil\lceil\log\log n\rceil\}.
\end{equation}
Maintain a common cap $e$, initially $n$. At a nonempty active population $a$, use
\[
 \kappa=\min\{a,2e\},\qquad q(\kappa)=\lceil nQ/\kappa\rceil.
\]
Probe at contender bound $\kappa$ with amplification $q(\kappa)\ge Q$. Our adaptive probe's cost is $\Oh(\kappa q(\kappa)\log N)=\Oh(B)$. The smaller contender bound permits higher accuracy at the same padded cost.

For an integer threshold $1\le D\le n$, procedure $\Prune(D)$ probes, saves inventories, and searches all keys for the maximum score $b$. If $b<D$, $\Prune(D)$ returns without changing $F$. Otherwise our pruning procedure broadcasts the winner's saved inventory, updates $F$ and $e\leftarrow\min\{e,b+1\}$, and repeats. An empty active set returns immediately. Dormant receivers remain eligible. The probe parameters and unconditional searches are independent of $D$, so all thresholds share prefixes.

A completed broadcast secures the credited sources and removes them from data contention. Every neighbor heard by the winner therefore leaves its active neighborhood. Aggregated rumors can secure further sources; the contraction needs only the heard neighbors.

\emph{Why the contraction takes two steps.} At fixed accuracy, a receiver may win logarithmically many times. Adaptive pruning also updates a nonincreasing contender cap from the winner's score. A smaller cap permits greater amplification at the same probe cost, so a win controls both current removal and future accuracy. The proof tracks pairs of wins because the new cap certificate is available only after the earlier selection. Intervening broadcasts only shrink the relevant sets.

Receivers with the same closed in-neighborhood can share this accounting even when different members win. The quantity that contracts is their common set of possible active origins. Lemma~\ref{lem:charge} combines this charge with the elementary budget on newly secured sources.

The quadratic estimate explains the $\log\log n$ charge. Dividing $x_{r+2}\le x_r^2/(4n)$ by $4n$ squares the normalized population $x_r/(4n)$, initially at most $1/4$, after each pair of wins. Here a qualification means an actual winning selection of the receiver or class. These wins need not occur in consecutive global iterations.

Put $H_\star=2\lceil\log_2(1+L)\rceil$ and $p_D(t)=\min\{H_\star,1+\lfloor\log_Q(t/D)\rfloor\}$ for $t\ge D$. Let $\mathcal S=\{\NB(v)\cup\{v\}:v\in V\}$ contain distinct closed in-neighborhoods, and put $t_D=|\{S\in\mathcal S:|S|>D\}|$,
\[
 \Psi_D=\sum_{S\in\mathcal S:\,|S|>D}p_D(|S|).
\]
Define $s_D=\lfloor n/D\rfloor$ and
\begin{align}
 \Phi_D&=\sum_{j=0}^{H_\star-1}h_{Q^jD}
       =\sum_{\indeg(v)\ge D}p_D(\indeg(v)),\label{eq:phi}\\
 \rho_D&=\min\{s_D,\Phi_D,\Psi_D\}.\label{eq:rho}
\end{align}
These are analytical parameters; neighborhoods and their classes are not inputs.
Exact probes retain their $N$ singleton slots. Put $\rho_D^{\rm ex}=\min\{s_D,t_D\}$; in particular $\rho_D^{\rm ex}\le\min\{s_D,h_D\}$.

The three terms in $\rho_D$ bound the same broadcasts: the source budget counts disjoint batches of size $D$, the receiver sum bounds repeated wins, and the class sum groups receivers sharing a closed neighborhood. Their minimum bounds one execution without requiring vertices to know which charge is smallest.

\begin{lemma}[Adaptive profile and neighborhood charging]\label{lem:charge}
Procedure $\Prune(D)$ makes at most $\rho_D$ successful selections, costs $\Oh((1+\rho_D)B)$ rounds, and returns with $d_A(v)<2D$ for all $v$. The charge holds for every qualifying prefix of the common execution. From an entry population $a$, the success cap is $\lfloor a/D\rfloor$. Moreover,
\begin{equation}\label{eq:coarsetail}
 \rho_D\le h_D\min\!\left\{H_\star,1+\frac{\log^+(\Delta/D)}{\log Q}\right\},
\end{equation}
and $\rho_D\le H_\star t_D$, where $\log^+x=\max\{0,\log x\}$. Exact pruning costs $\Oh((1+\rho_D^{\rm ex})(B+N))$.
\end{lemma}
\begin{proof}
\emph{Invariant and source budget.} Since degrees and $e$ are integers, the invariant $d_A(v)<2e$ bounds every closed contender set by $\min\{a,2e\}$. This invariant holds at $e=n$. A maximum score $b$ gives $d_A(v)\le b/(1-1/Q)\le2b$ everywhere, including dormant receivers. Degrees only decrease after broadcasting, so $e\leftarrow\min(e,b+1)$ preserves the invariant. Each success secures at least $D$ new origins, giving $\lfloor a/D\rfloor$ selections. The winning receiver's active degree decreases by a factor of at least $q(\kappa)\ge Q$.

\emph{Receiver charge and persistence of the cap.} Index successful global iterations by $t$, with $A_t,e_t$ denoting the state before the probe and $b_t$ its maximum score. Fix a receiver $v$, and let $t_1<t_2<\cdots$ be the iterations in which $v$ is the selected winner. Only these wins are charged to $v$. Put $x_r=d_{A_{t_r}}(v)$ and let $q_r$ be that probe's amplification. Since the winning score satisfies $D\le b_{t_r}\le x_r$, its update installs
\[
 e_{t_r+1}=\min\{e_{t_r},b_{t_r}+1\}\le x_r+1.
\]
Every subsequent update also takes a minimum. Thus, regardless of which receivers win between $t_r$ and $t_{r+1}$,
\[
 e_{t_{r+1}}\le e_{t_r+1}\le x_r+1,
 \qquad q_{r+1}\ge\frac{nQ}{2(x_r+1)}.
\]
A larger intervening score cannot increase the cap. The broadcast at $t_r$ leaves at most $x_r/q_r$ active neighbors of $v$, and later broadcasts only shrink that set. Hence $x_{r+1}\le x_r/q_r\le x_r/Q$. Whenever the indicated wins exist,
\[
 x_{r+2}\le\frac{x_{r+1}}{q_{r+1}}
 \le\frac{2x_r(x_r+1)}{nQ^2}
 \le\frac{x_r^2}{4n},
\]
where the last inequality uses $x_r\ge1$ and $Q\ge4$.

Starting from $x_1\le n$, induction gives $x_{2t+1}\le n4^{1-2^t}$. For $t=H_\star/2=\lceil\log_2(1+L)\rceil$, this is at most $n4^{-L}=1/n<1$. Since each qualifying win requires $x_r\ge D\ge1$, there are at most $H_\star$ wins. The one-step $Q$-contraction also bounds their number by $1+\lfloor\log_Q(\indeg(v)/D)\rfloor$. Summing the smaller charge over receivers proves $\Phi_D$ and~\eqref{eq:coarsetail}.

\emph{Neighborhood-class charge.} Fix a closed neighborhood $S$, and now index only iterations won by receivers whose original closed neighborhood is $S$. Let $x_r=|A\cap S|$ before the $r$th such win, and let $I_r$ indicate whether its winner is active. The winner has active in-degree $x_r-I_r$. Its snapshot secures its own source and all heard neighbors, leaving at most $(x_r-I_r)/q_r\le x_r/q_r$ active sources in $S$. Its score also satisfies $b\le x_r-I_r\le x_r$, so the update installs $e\le x_r+1$. The cap remains below this value through every intervening win. The same two-step and geometric recurrences therefore apply even when successive winners in the class are different vertices. Only classes with $|S|>D$ can win, giving $\Psi_D$ and $H_\star t_D$. The source, receiver, and class charges count the same successes, so their minimum is $\rho_D$.

\emph{Stopping certificate.} At the stopping search, every integer score is at most $D-1$, so $d_A(v)\le(D-1)/(1-1/Q)<2D$. Every probe and selection costs $\Oh(B)$. An exact winning snapshot secures the winning receiver's entire closed neighborhood; that class cannot qualify again. Thus exact probes have charge $t_D$ and retain their $N$-slot cost.
\end{proof}

\subsection{Why closed-neighborhood classes can share a charge}
The receiver's own rumor makes the closed neighborhood essential. With $x$ active origins and winner-activity indicator $I\in\{0,1\}$, the saved snapshot secures at least $(1-1/q)(x-I)+I$ class members. Thus the charge covers active and dormant winners without requiring the protocol to identify their class. The common cap retains the bound installed at each win.

Intervening selections only decrease the population and cap. Thus the paired recurrence applies across different winners in the same class. Starting from $x_r\le n$, successive paired bounds are $n/4$, $n/64$, and $n/16384$. The resulting $\Oh(\ell)$ qualification cap improves the fixed-accuracy $\Oh(L/\ell)$ multiplicity; grouping then counts distinct closed neighborhoods instead of individual receivers.

\subsection{Monotonicity and the baseline execution}
The parameters $\Phi_D,\Psi_D$ and $t_D$ are nonincreasing in $D$, supporting the threshold-prefix and rounding arguments below. The outer controller preserves its current cap $e$. A separate pruning call may initialize its cap to $n$; within each call the cap never increases, and Lemma~\ref{lem:charge} applies to that call's own trace. Exact probes, and the one-time $Q=2$ degree estimates used later, remain valid separate operations. A baseline copy keeps fixed amplification $Q$, $n$-bounded probes, and independent clocks and knowledge. Give the baseline copy and the baseline's pruning--path-cover pool constant shares. The baseline's untruncated charge is $\sum_{j\ge0}h_{Q^jD}$. This preserves the fixed-accuracy guarantees on their own trace: increased accuracy can change winners, so a residual witness is not transferred between executions.

\section{Completion and Exact Temporal Replay}
The certified interference bound now supplies a completion procedure. A rumor is \emph{quasi-disseminated} if it is globally known or held by a vertex whose own original rumor is globally known~\cite{grx}. This condition is monotone. We first establish quasi-gossiping, then use exact replay to deliver every rumor everywhere (Lemma~\ref{lem:replay}).

\emph{Why collection suffices.} After $q$ spreading periods along a path, its first $q+1$ vertices hold the target, and the endpoint knows their original rumors. Collection leaves fewer than $q+1$ of those origins unsecured. Hence one target holder has its own rumor secured. The common endpoint inventory supplies the needed connection between holders and secured origins; a holder count alone would not establish it.

\begin{lemma}[Bounded-active-degree completion]\label{lem:complete}
Given common $F$ and integer $1\le K\le n$ with $d_A(v)\le K$ for all $v$, quasi-gossiping is reached at a common boundary in
\[
 \Oh(nK L^{3/2}\ell^{1/2}+B)
\]
additional rounds. If $K=0$, quasi-gossiping already holds.
\end{lemma}
\begin{proof}
Freeze $A$ during the data phase. A strong selector for $\min(n,K+1)$ contenders has length $\sigma=\Oh((K+1)^2\log N)$. The closed contender set at a receiver is $A\cap(\NB(v)\cup\{v\})$, so one period delivers every active sender's entry knowledge to all its out-neighbors, with the receiver listening. Run $q$ periods, then freeze $H_v=K_v$ and execute $\Drain(q+1)$. For $1\le q<n-1$ the cost is at most
\begin{equation}\label{eq:completioncost}
 q\sigma+B\left(1+\frac{n}{q+1}\right).
\end{equation}
For $q=n-1$, omit collection.

Fix a rumor $r$ not quasi-disseminated when the data phase begins; the rumor's source is active. If there is a dormant vertex, choose a simple path from the source to a first dormant vertex. A path of at most $q$ edges delivers $r$ to that vertex. Otherwise the path's first $q+1$ vertices $v_0,\ldots,v_q$ are active. If there is no dormant vertex, either a simple path of length $q$ starts at the source, or every vertex has a source-to-vertex path of fewer than $q$ edges and receives $r$ during the data phase.

In the remaining path case, each $v_i$ holds $r$ after $q$ periods. Also every original rumor of $v_0,\ldots,v_q$ reaches $v_q$ along the corresponding path suffix, so all $q+1$ sources belong to $H_{v_q}$. After collection, $|H_{v_q}\setminus F|<q+1$, so at least one of these vertices is secured. That vertex retains $r$, establishing quasi-dissemination. This argument applies to all rumors simultaneously. At $q=n-1$, a simple path to every target or to a first dormant vertex suffices without collection. For $K=0$, strong connectivity forces $A=\varnothing$, since any active source would have an out-neighbor with positive active in-degree.

Choose $q=\min\{n-1,\max\{1,\lceil\sqrt{Bn/\sigma}\rceil\}\}$. Interior balancing gives $\Oh(B+\sigma+\sqrt{Bn\sigma})$. At the lower endpoint $Bn\le\sigma$ gives the same bound. At the upper clipping point $\sqrt{Bn/\sigma}>n-2$; omitting collection costs $\Oh(\sqrt{Bn\sigma})$ for $n\ge4$. Finally $\sigma=\Oh(nK\log N)$ and $\sqrt{Bn\sigma}=\Oh(K\sqrt{Bn\log N})$ yield the result.
\end{proof}

\subsection{Endpoint choices in completion}
The integer $q$ counts whole selector periods. At the upper endpoint, collection is unnecessary: each rumor reaches all vertices or a first dormant holder. Strong connectivity supplies the latter route with active internal vertices even when the active induced graph is disconnected. The rounding estimates above cover both clipped endpoints; for $n<4$, the strong-selector convention in Section~II applies.

\begin{lemma}[Exact replay]\label{lem:replay}
Start with one original rumor per vertex and record every transmit/listen action of a full-knowledge execution of length $T$. If quasi-gossiping holds at the first pass's end, repeating those exact actions for another $T$ rounds, while retaining all knowledge, completes gossiping.
\end{lemma}
\begin{proof}
Let $M_{su}=1$ exactly when a first-pass temporal path, allowing waiting, carries $s$'s original rumor to $u$. Fix a source--destination pair $(s,v)$. If $s$ is globally known, choose $u=v$; retention supplies the self-path at that destination. Otherwise quasi-gossiping supplies a holder $u$ whose own rumor is globally known. In both cases $M_{su}=M_{uv}=1$.

At replay entry, $u$ already holds $s$'s rumor. Repeating the recorded actions preserves all listening and collision patterns, so the first-pass path that carried $u$'s rumor to $v$ now also delivers $s$'s rumor. The first pass establishes a useful holder; replay supplies the holder's dissemination path in the required temporal order. Thus the Boolean square $M^2$ is all ones, proving gossiping.
\end{proof}

\subsection{Temporal closure and complete recording}
A temporal path follows successful receptions in chronological order, allowing waiting. Full-knowledge packets make temporal reachability equivalent to rumor knowledge; retention supplies the diagonal of $M$. The useful holder joins a source to a destination across two passes, without reversing any arc.

For example, a holder may broadcast its own rumor before learning a target. Its original broadcast then lacks the target, although the target is quasi-disseminated at the first pass's end. Replay starts with the target already at that holder and reuses the earlier dissemination path. Recording must therefore include the entire execution from the original state, including control. Replay repeats actions, not selection predicates, packet gates, or stopping decisions.

Protocol $\Gossip(D)$ initializes $F=\varnothing$ and starts recording, runs $\Prune(D)$, invokes Lemma~\ref{lem:complete} with $K=\min(n,2D)$, then replays the entire first pass.

\section{Profile Bound and Adaptation}
A fixed threshold gives a performance benchmark: pruning pays for scores at least $D$, and completion pays for the certified residual interference. The adaptive protocol attains the best benchmark without learning the degree profile. The connection is a shared greedy prefix: all thresholds use the same probes, frozen scores, and unconditional key searches until their stopping comparisons differ. A common completion budget then determines the handoff, allowing one evolving execution to be compared with every threshold.

\begin{theorem}[Degree-profile bound]\label{thm:profile}
For every common integer $1\le D\le n$, $\Gossip(D)$ completes gossiping in $\Oh(G_D)$ rounds, where
\begin{equation}\label{eq:GD}
 G_D=nD L^{3/2}\ell^{1/2}+n(1+\rho_D)L^2\ell.
\end{equation}
Only $n,N,D$ are inputs. Define $G_D^{\rm ex}$ by replacing $\rho_D$ with $\rho_D^{\rm ex}$. Exact probes take $\Oh(G_D^{\rm ex}+N(1+\rho_D^{\rm ex}))$ rounds, hence $\Oh(G_D^{\rm ex})$ if $N=\Oh(nL^2\ell)$.
\end{theorem}
\begin{proof}
Lemma~\ref{lem:charge} bounds pruning and certifies $d_A<2D$. Lemma~\ref{lem:complete} bounds quasi-gossiping completion, and Lemma~\ref{lem:replay} doubles the first pass. Exact pruning costs $\Oh((1+\rho_D^{\rm ex})(B+N))$; completion and replay are unchanged. Since $B=\Theta(nL^2\ell)$, this gives the stated exact bound.
\end{proof}

\begin{theorem}[Adaptation without a threshold-search penalty]\label{thm:adaptive}
Using only common $n,N$, deterministic gossiping takes
\[
 \Oh\!\left(\min_{1\le D\le n}\min\{G_D,G_D^{\rm ex}+N(1+\rho_D^{\rm ex})\}\right)
\]
rounds. This includes $\Oh(\min_DG_D^{\rm ex})$ when $N=\Oh(nL^2\ell)$. Our adaptive protocol retains the $\Ot(n^{4/3})$ general guarantee of~\cite{kuschner} and, for linear identifiers, $\Ot((mn)^{3/5})$ from~\cite{wu}, without knowing $m$. Further branches attain Theorems~\ref{thm:density} and~\ref{thm:regular} without structural tests.
\end{theorem}
\begin{proof}
Pad each probe--$\MaxSelect$ iteration to a common length $U=\Theta(B)$. For a sufficiently large fixed constant $c$, let
\begin{equation}\label{eq:CD}
 C_D=c\left(B+D\lceil\sqrt{Bn\log N}\rceil+nD\log N\right).
\end{equation}
This is nondecreasing in $D$ and bounds Lemma~\ref{lem:complete} with $K=\min(n,2D)$. This budget is $\Theta(B+nD L^{3/2}\ell^{1/2})$.

Initialize $F=\varnothing$, $P=0$, $E=n$, and recording. Each iteration probes using the common current cap $E$ and executes $\MaxSelect(b_v)$, including the winner's inventory broadcast. All vertices learn $b=\max_vb_v$. Set
\[
 P\leftarrow P+U,\qquad E\leftarrow\min(E,b+1).
\]
Stop if $|F|=n$, $b=0$, or $P\ge C_E$. Complete any nonempty active set with $K=\min(n,2E)$, then replay the entire first pass.

Before broadcasting, $d_A(v)\le2b$ by~\eqref{eq:probe}; afterward degrees only decrease. Thus $d_A(v)<2E$ persists. Each positive-score iteration secures a new source, so termination occurs within $n+1$ iterations. Completion and replay are valid, also with exact probes.

Fix any comparator $D$. Before the first $b<D$, Lemma~\ref{lem:charge} bounds the prefix by $\rho_D$ iterations. If our controller stops in that prefix, $P\le U\rho_D$; completion costs at most $P$, or is unnecessary when $F=V$. Otherwise the first $b<D$ iteration has $P\le U(\rho_D+1)$ and makes $E\le D$. Any later iteration is entered with $P<C_E\le C_D$; after that iteration, $P<C_D+U$. Consequently the eventual pruning time is at most
\[
 \max\{U(\rho_D+1),C_D+U\},
\]
and completion costs at most $C_D$. The total is $\Oh(G_D)$, uniformly in $D$. Replay doubles the first-pass time. For exact probes, use $U_{\rm ex}=\Theta(B+N)$ and the same $C_D$. The comparison gives $\Oh(G_D^{\rm ex}+N(1+\rho_D^{\rm ex}))$, which is $\Oh(G_D^{\rm ex})$ if $N=\Oh(B)$.

Privately interleave both profile variants, the path branch of Theorem~\ref{thm:pathadaptive}, Theorems~\ref{thm:density} and~\ref{thm:regular}, and the general branch~\cite{kuschner}. For linear identifiers, retain~\cite{wu}, removing the earlier protocol's edge-count input by verified geometric guesses as below. Give the profile branches a constant share of time; the pruned-density pool has its own share. Our protocol stops at the first common completion.
\end{proof}

\section{Neighborhood Clearing and an Improved Profile Bound}\label{sec:clearing}
Adaptive accuracy limits how often a receiver can qualify, but each qualification still pays for a global election. We now retain an elected receiver and reuse it for subsequent inventory broadcasts. A broadcast from this known source costs $R$, while an election costs $B=\Theta(R\log n)$. Thus $\Oh(\log n)$ probe-and-broadcast repetitions fit within $\Oh(B)$ rounds and can secure the receiver's entire closed in-neighborhood.

Clearing has private inventories, a secured ledger, clocks, and recording. The later survivor and reception-cover guarantees retain the reception-score execution of Section~\ref{sec:reception}.

Election makes the receiver known; later broadcasts use its fresh snapshots without another search. Secured receivers still receive data and relay control.

\begin{lemma}[Clearing an elected receiver]\label{lem:clearreceiver}
Suppose $F$ is common and secured and the identifier of a present vertex $v$ is common. A deterministic procedure $\Clear(v)$ uses only $n,N,v,F$ and, in $\Oh(B)$ rounds, produces a common secured ledger $F'$ satisfying
\begin{equation}\label{eq:clearsound}
 F\cup\NB(v)\cup\{v\}\subseteq F'.
\end{equation}
It need not discover $\NB(v)$ or its size.
\end{lemma}
\begin{proof}
Set $t=\lceil\log_4 n\rceil$. Each of $t$ repetitions has the following common order. First freeze the current sender set $A=V\setminus F$ and run Lemma~\ref{lem:probe} with contender bound $n$ and amplification $4$. Then $v$ saves its actual inventory $H_v=K_v$ and initiates an $R$-round broadcast of that saved inventory. At the completed broadcast boundary every vertex updates $F\leftarrow F\cup H_v$. The next repetition uses this refreshed active set. The known source $v$ needs no new election. Control relays use the full physical network, including vertices whose sources are secured.

For $i=0,\ldots,t-1$, let $A_i$ be the active set at repetition $i$ and $d_i=|\NB(v)\cap A_i|$; let $d_t$ denote the degree after the last update. The probe hears at least $3d_i/4$ distinct active in-neighbors. Every heard packet contains its sender's own original rumor, so all those source identifiers belong to $H_v$. The completed broadcast secures them, and therefore
\begin{equation}\label{eq:clearcontraction}
 d_{i+1}\le\lfloor d_i/4\rfloor.
\end{equation}
The saved inventory also contains $v$'s own rumor, and any additional genuine knowledge can only strengthen the ledger update. Iterating~\eqref{eq:clearcontraction} gives $d_t\le(n-1)/4^t<1$, hence $d_t=0$. Every original in-neighbor is therefore secured, as is $v$ itself after the first broadcast. Thus repeated certified removals establish~\eqref{eq:clearsound}.

The probe takes $\Oh(n\log N)$ rounds. Each known-source broadcast takes $R$, so the total is
\[
 \Oh\!\left(t[n\log N+R]\right)=\Oh(R\log n)=\Oh(B).
\]
All decisions use prescribed common boundaries. Each invocation resets its private control state and clock. Only the designated saved inventory receives ledger credit, even when transmitted packets append additional knowledge. If $F=V$ earlier, the remaining prescribed time may be idle.
\end{proof}

For the original closed neighborhoods $S_v=\NB(v)\cup\{v\}$, recall that $t_D$ counts distinct sets $S_v$ with $|S_v|>D$. Define
\begin{equation}\label{eq:clearcharge}
 \chi_D=\min\{\lfloor n/D\rfloor,t_D\}.
\end{equation}
Then $\chi_D\le t_D\le h_D$. Numerically, $\chi_D=\rho_D^{\rm ex}$: clearing attains the exact-pruning class charge using selective schedules. The two procedures can nevertheless produce different ledgers.

For fixed public $D$, procedure $\CPrune(D)$ repeats the following phase. Probe at contender bound $n$ with amplification $4$, freeze scores and inventories, and search all keys for the maximum score $b$. If $b<D$, return before broadcasting a snapshot. Otherwise broadcast the winner's saved inventory, update $F$, and call $\Clear$ on that winner. Return immediately when $F=V$. The probe and key search are independent of $D$, so all thresholds share complete phases until their stopping comparisons differ.

\begin{lemma}[Complete-neighborhood charging]\label{lem:clearcharge}
Procedure $\CPrune(D)$ has at most $\chi_D$ successful clearing phases, takes $\Oh((1+\chi_D)B)$ rounds, and returns with $d_A(w)<2D$ for every physical receiver $w$. The charge applies to every prefix whose phase-entry maxima are at least $D$. From an entry population $a$, the source-budget part is $\lfloor a/D\rfloor$.
\end{lemma}
\begin{proof}
If $b\ge D$, the winner's saved snapshot contains the original rumors of at least $D$ distinct sources active at that phase's entry. Those sources are secured by the first completed inventory broadcast. The winner itself need not be active: its own source may already belong to $F$. The charge uses only the $D$ heard active neighbors, whose source tags are outside the entry ledger. It therefore applies equally to active and dormant winners.

At the next phase, every source just credited is absent from the refreshed sender set. Its rumor may still be carried in aggregate packets, but it cannot contribute another immediate-sender identity. Thus the batches charged to first broadcasts are disjoint across all phases, even if inventories overlap or control has disseminated additional uncredited rumors. Charging at least $D$ new ledger entries per phase gives $\lfloor a/D\rfloor$ successes from an entry population $a$, and $\lfloor n/D\rfloor$ from the original state.

For the class budget, a winner with score at least $D$ has original in-degree at least $D$, so its original closed neighborhood $S$ has size greater than $D$. Lemma~\ref{lem:clearreceiver} secures every source in $S$ before this phase ends. Every other receiver with the same closed neighborhood thereafter has active in-degree zero. Thus this class can never justify another positive-score phase. At most $t_D$ classes can be charged. The source and class budgets count the same phases, proving~\eqref{eq:clearcharge}.

At the stopping search, $b_w\le D-1$ for every $w$. The amplification-$4$ certificate gives
\[
 d_A(w)\le\frac{4}{3}b_w\le\frac{4}{3}(D-1)<2D.
\]
Every successful phase, including its first probe, election, inventory broadcast, and complete clearing, costs $\Oh(B)$. The final probe and maximum-key search also cost $\Oh(B)$. No vertex tests equality of neighborhoods, and no step scans the identifier universe.
\end{proof}

Table~\ref{tab:clearingcost} accounts for the whole phase: the initial selected snapshot supplies at least $D$ new origins, and subsequent known-source broadcasts secure the entire neighborhood.

\begin{table}[t]
\centering\small
\caption{Cost of a positive-score clearing phase. Set $t=\lceil\log_4 n\rceil$. All ledger updates occur after their designated broadcasts complete.}\label{tab:clearingcost}
\renewcommand{\arraystretch}{1.08}
\begin{tabular}{@{}>{\raggedright\arraybackslash}p{0.56\columnwidth}>{\raggedright\arraybackslash}p{0.35\columnwidth}@{}}
\toprule
Operation & Round budget\\
\midrule
Probe; freeze scores and snapshots & $\Oh(n\log N)$\\
Search all keys for one winner & $\Oh(R\log n)$\\
Broadcast the selected snapshot & $R$\\
$t$ fresh probes at the known receiver & $\Oh(tn\log N)$\\
$t$ broadcasts from that receiver & $tR$\\
\midrule
Complete phase, with common padding & $U=\Theta(B)$\\
\bottomrule
\end{tabular}
\end{table}

A successful phase uses its score to charge new sources and a neighborhood class. A stopping search instead certifies residual interference at every receiver. It ends before the inventory broadcast because completion needs only that certificate.

\begin{theorem}[Clearing-profile adaptation]\label{thm:clearadaptive}
For every polynomial identifier range, one deterministic protocol using only common $n,N$ completes gossiping in $T$ rounds, where
\begin{equation}\label{eq:clearprofile}
 \begin{aligned}
 T&=\Oh\!\left(\min_{1\le D\le n}\widehat G_D\right),\\
 \widehat G_D&=nDL^{3/2}\ell^{1/2}+n(1+\chi_D)L^2\ell.
 \end{aligned}
\end{equation}
In particular, it attains the numerical exact-pruning profile budget of Theorem~\ref{thm:profile} without an $N$-slot probe and without taking $D$, degrees, or neighborhood classes as inputs.
\end{theorem}
\begin{proof}
\emph{Fixed-threshold benchmark.} For fixed $D$, Lemma~\ref{lem:clearcharge} followed by Lemma~\ref{lem:complete} with $K=\min(n,2D)$ reaches quasi-gossiping in $\Oh(\widehat G_D)$ rounds. Record from the original rumor state and replay every action of this entire first pass, including election and clearing control. Lemma~\ref{lem:replay} then gives full gossiping at twice the first-pass cost.

\emph{Controller.} Pad each complete successful clearing phase to a common length $U=\Theta(B)$. Choose its constant large enough for the initial probe, unconditional maximum selection, and all $t$ repetitions of $\Clear$. Let $C_e$ be the nondecreasing completion budget in~\eqref{eq:CD}. It bounds Lemma~\ref{lem:complete} with $K=\min(n,2e)$.

Initialize $F=\varnothing$, $P=0$, $e=n$, and the recording. At the start of each phase, return if $F=V$. Otherwise probe with amplification $4$ and contender bound $n$, freeze all scores and snapshots, and run the unconditional maximum-key search. This search ends before the selected inventory is broadcast. If the maximum $b=0$, then $A=\varnothing$: otherwise an active source has an out-neighbor by strong connectivity, and the positive-degree reception certificate gives a positive score somewhere. For $b>0$, broadcast the selected snapshot, run $\Clear$ on its source, and finish the padding to $U$ rounds. Then set
\[
 P\leftarrow P+U,\qquad e\leftarrow\min\{e,b+1\}.
\]
Stop if $F=V$ or $P\ge C_e$. Otherwise begin the next phase. At a nonterminal handoff, invoke Lemma~\ref{lem:complete} and then replay the complete first pass.

\emph{Invariant and termination.} The phase-entry maximum satisfies $d_A(w)\le4b/3<2(b+1)$ everywhere. All subsequent ledger updates in that phase only shrink $A$. Together with the previous cap this proves the invariant $d_A(w)<2e$ at every phase boundary. Thus completion uses a valid bound. Every positive phase secures a new source, so the reference execution that ignores the handoff rule has at most $n$ phases. Its probes, searches, clearing repetitions, and padding are independent of any analytical threshold.

\emph{Comparison with a threshold.} Fix $D$. Before the reference execution's first maximum below $D$, Lemma~\ref{lem:clearcharge} allows at most $\chi_D$ phases. If the actual controller stops within that prefix, its pruning cost is at most $U\chi_D$; its completion budget is at most the elapsed cost by the stopping test, unless $F=V$ and completion is unnecessary. This case costs $\Oh((1+\chi_D)U)$ including replay. A reference execution that terminates in that prefix is covered by the same estimate.

Otherwise the first below-$D$ maximum begins phase $j\le\chi_D+1$. If the active set is nonempty, completing that phase makes $e\le b+1\le D$, while $P\le U(\chi_D+1)$. Every later entered phase starts with $P<C_e\le C_D$ and adds exactly $U$. Consequently the eventual pruning cost is at most
\[
 \max\{U(\chi_D+1),C_D+U\},
\]
and the completion cost is at most $C_D$. The zero-score or fully secured cases require no completion and satisfy the same bound. Hence the total, including full replay, is $\Oh((1+\chi_D)B+C_D)=\Oh(\widehat G_D)$. This holds for every comparator $D$ on the same execution, proving~\eqref{eq:clearprofile} without a threshold pool.
\end{proof}

\subsection{Comparison with every reference prefix}
The controller also admits a direct comparison with every boundary of its reference execution. Let $e_j$ and $F_j$ be the common state after $j$ complete phases, with $j=0$ allowed. Put $C_j=C_{e_j}$ when $F_j\ne V$ and $C_j=0$ at a fully secured state. These budgets are nonincreasing. If the actual controller stops at phase $i\le j$, its completion is paid by $iU$, or is unnecessary when $F_i=V$. Including replay gives $T\le4jU$. If $i>j$, the preceding stopping test gives
\[
 (i-1)U<C_{i-1}\le C_j.
\]
The final phase adds $U$ and completion costs at most $C_j$, so $T<4C_j+2U$. Consequently,
\begin{equation}\label{eq:clearprefix}
 T\le4\min_j\max\{jU,C_j\}+2U.
\end{equation}
Reference states beyond the actual handoff are defined by continuing the same prescribed phases. They are analytical comparators only; the additive $2U$ pays for crossing a completion budget by one complete phase.

For a threshold $D$, the first below-$D$ phase supplies a comparator with $j\le\chi_D+1$ and $e_j\le D$, unless the reference has already secured all sources. Substitution in~\eqref{eq:clearprefix} recovers Theorem~\ref{thm:clearadaptive}. Thus the threshold analysis is a specialization of one comparison between increasing elapsed time and decreasing completion cost.

\subsection{Relation to the earlier profile bound}
The numerical profiles clarify the improvement. For every integer $D$,
\begin{equation}\label{eq:compareprofiles}
 \chi_D\le\rho_D\le H_\star\chi_D,
 \qquad H_\star=\Oh(\ell).
\end{equation}
Every contributing weight $p_D$ is at least one, so $\Phi_D\ge h_D\ge t_D$ and $\Psi_D\ge t_D$. Taking the minimum with $s_D=\lfloor n/D\rfloor$ proves the first inequality. For the second,
\[
 \rho_D\le\min\{s_D,H_\star t_D\}
 \le H_\star\min\{s_D,t_D\}=H_\star\chi_D.
\]
These are algebraic comparisons of the defined profiles. Lemma~\ref{lem:charge} separately establishes that $\rho_D$ bounds actual successes, using the persistent cap along each winner subsequence.

The completion terms of $G_D$ and $\widehat G_D$ are identical. Equation~\eqref{eq:compareprofiles} therefore compares their complete degree-profile budgets:
\[
 \widehat G_D\le G_D\le H_\star\widehat G_D.
\]
Adaptive accuracy allows a neighborhood class to qualify at most $H_\star$ times; clearing secures that whole class within one phase of the same asymptotic cost. Under~\eqref{eq:concentration}, the earlier charge gives $\rho_D=\Oh(\dbar\ell)$ at $D=\lceil C_0\dbar\rceil$, leading to $\Oh(mL^2\ell^2)$. Clearing replaces this with $\chi_D=\Oh(\dbar)$ and gives $\Oh(mL^2\ell)$. If the chosen threshold exceeds $n$, taking $D=n$ gives the same bounds because $m=\Omega(n^2)$.

The inequalities compare budgets on different executions. Reception-score pruning may leave a small population or a useful cover, so both procedures remain in the portfolio. A constant share preserves~\eqref{eq:clearprofile}; each completion and replay uses its own ledger. The following sections develop these complementary residual guarantees.

\section{Path Selectors and Surviving Sources}\label{sec:path}
Completion can exploit both residual interference and the number $a=|V\setminus F|$ of unsecured sources. Any rumor held by a dormant vertex is already quasi-disseminated. Path selectors~\cite{grx} move the remaining rumors along routes whose combined interference fits one schedule scale. The procedures below apply to any common secured ledger; Section~\ref{sec:adaptivepath} combines them with reception-score pruning.

For a fixed data-carrier set $C\subseteq V$, a carrier path $P=(v_0,\ldots,v_t)$ is simple and has all vertices except possibly $v_t$ in $C$. The path's closed carrier neighborhood is
\begin{equation}\label{eq:pathnb}
 Y_C(P)=C\cap\bigcup_{i=1}^{t}(\NB(v_i)\cup\{v_i\}).
\end{equation}
Only $C$ sends data; other vertices listen. The active-path case uses $C=A$.

Path selection isolates senders in chronological order, carrying a rumor along successive edges in one invocation. The scale must cover both outside interfering carriers and active receivers, which are included in $Y_C(P)$. Fixing $C$ preserves this neighborhood throughout the invocation.

Even a two-edge path can require a large scale if many other carriers send into its receivers. Conversely, a small closed carrier neighborhood permits an ordered subsequence that carries knowledge along the whole path. Vertices execute the common schedule without discovering the path or subsequence. Since packets carry all current knowledge, the selector's guarantees hold simultaneously for all such paths.

\begin{lemma}[Path selection with fixed carriers]\label{lem:path}
Freeze $C$. For $1\le k\le n$, a schedule of length
\begin{equation}\label{eq:pathlatency}
 S_k=\Oh(k^2\log N)
\end{equation}
delivers all entry rumors from $Y_C(P)$ to the endpoint of every carrier path with $|Y_C(P)|\le k$. Only $C$ sends data.
\end{lemma}
\begin{proof}
By~\cite[Theorem 1]{kuschner}, a permutation selector of size $\Oh(k^2\log N)$ isolates every $k$-set in each prescribed order. This implements path selection~\cite{grx}. Fix $Y=Y_C(P)$ and pad $Y$ to a $k$-set of identifiers. Extend the distinct path senders $v_0,\ldots,v_{t-1}$, in that order, to a permutation of this set. The selector isolates them in successive slots. Every carrier in-neighbor of the next receiver belongs to $Y$; a carrier receiver itself also belongs to $Y$. Consequently each designated reception succeeds in the half-duplex model, carrying all retained entry rumors forward.

For any $u\in Y$, either $u$ lies on $P$ or $u$ has an arc into $P$. Prepend that arc to the appropriate suffix, or start at $u$ when $u$ already lies on the suffix. The resulting simple carrier path has neighborhood contained in $Y$. Applying the same argument delivers the source's entry rumors to the endpoint. All guarantees hold simultaneously because the selector handles every permutation and packets retain all values. For $k=1$, one slot suffices.
\end{proof}

A small surviving population favors direct path selection. Small contender sets favor repeated spreading followed by collection. The third option combines further pruning with path doubling. All successful collections secure disjoint members of the same completion-entry population, bounding their total cost.

\begin{lemma}[Completion sensitive to surviving sources]\label{lem:survivors}
Suppose $F$ is common and secured, $a=n-|F|$, and a common integer $1\le K\le a$ bounds every closed active contender set. Quasi-gossiping can be completed with no identifier round robin in $\Oh(B+\mathcal C(a,K))$ rounds, with second-pruning budget
\[
 \mathcal J(a)=(Ba)^{2/3}[\Lambda\log(a+1)]^{1/3}.
\]
The completion menu is
\begin{equation}\label{eq:completionmenu}
 \mathcal C(a,K)=\min\{a^2\Lambda,K\sqrt{Ba\Lambda},\mathcal J(a)\}.
\end{equation}
If $a=0$, gossiping is already complete.
\end{lemma}
\begin{proof}
\emph{Direct path selection.} Run Lemma~\ref{lem:path} with $k=a$. If $F\ne\varnothing$, every active source has a simple path to a first dormant vertex, with all internal vertices active and neighborhood contained in $A$. The source's rumor reaches that dormant holder and is quasi-disseminated. If $F=\varnothing$, every simple source--destination path has neighborhood within $A$, so gossiping itself finishes.

\emph{Spreading and collection.} Run $q$ strong-selector periods of length $\sigma=\Oh(K^2\Lambda)$. Each rumor is then globally known, held by a dormant vertex, or held at $q+1$ active vertices on a simple source path. The latter vertices' original rumors reach the endpoint. Freeze $H_v=K_v$ and run $\Drain(q+1)$, securing a holder and thus quasi-disseminating the rumor. For $q=a$ or $q\ge n-1$, omit collection. The cost is
\[
 \Oh(q\sigma+B(1+a/(q+1))).
\]
Set $q=\min\{a,\lceil\sqrt{Ba/\sigma}\rceil\}$. Then $q\sigma\le\sqrt{Ba\sigma}+\sigma$; if collection is used, $Ba/(q+1)\le\sqrt{Ba\sigma}$. Since $K^2\le na$ and $B=\Omega(n\Lambda)$, $\sigma=\Oh(Ba)$, so the rounding term is absorbed in $\Oh(B+K\sqrt{Ba\Lambda})$.

\emph{Second pruning and path doubling.} Put $g=\Lambda\log(a+1)$ and choose
\[
 h=\left\lceil(Ba/g)^{1/3}\right\rceil,
 \qquad k=\text{first power of two at least }4h.
\]
Since $\log_2(a+1)\le a$ and $B=\Omega(\Lambda)$, we have $g=\Oh(Ba)$ and $h=\Theta((Ba/g)^{1/3})$. Dyadic rounding gives $4h\le k<8h$. If $k\ge a$, then $a=\Oh((Ba/g)^{1/3})$, so direct path selection costs $\Oh(a^2g)=\Oh(\mathcal J(a))$.

Otherwise run a second $\Prune(h)$ from the current secured ledger. Our second pruning stage has at most $\lfloor a/h\rfloor$ successful broadcasts, because each secures $h$ distinct sources from this entry population. The stage's final search certifies closed contender sets of size at most $2h<k$. Let $a'\le a$ sources remain. If $a'\le k$, direct path selection costs $\Oh(k^2g)$. Otherwise repeat path selection at scale $k$, followed by $\Drain(k/2)$.

Each doubling iteration alternates a path pass with collection. The pass gathers each covered entry-neighborhood inventory at an endpoint; collection leaves fewer than half a selector scale of those origins unsecured. Two residual neighborhoods therefore fit within one selector scale. The carrier-path argument below also permits a dormant endpoint.

\emph{Doubling invariant.} Let $A_i$ be the carriers before pass $i$, starting with $i=0$ after second pruning. Freeze $A_i$ during that pass. Following~\cite[Section 3.2]{grx}, we prove by induction that every $A_i$-carrier path of at most $2^i$ edges has neighborhood size below $k$; its endpoint may be dormant. The base case follows from the one-edge bound $2h<k$.

After the path pass and $\Drain(k/2)$, let $A_{i+1}\subseteq A_i$ be the remaining carriers. Split any $A_{i+1}$-carrier path $P$ of at most $2^{i+1}$ edges into at most two $A_i$-carrier paths $P_1,P_2$ of at most $2^i$ edges. The pass delivers each piece's entry-neighborhood inventory to its endpoint. The subsequent frozen collection gives
\[
 |A_{i+1}\cap Y_{A_i}(P_j)|<k/2
 \quad\text{for each nonempty piece }P_j.
\]
Their residual neighborhoods unite to $Y_{A_{i+1}}(P)$, whose size is below $k$. This proves the induction step.

At $i=\lceil\log_2 k\rceil$, a surviving length-$k$ path would contradict the invariant because its $k$ distinct senders belong to its closed carrier neighborhood. Thus every remaining carrier path has fewer than $k$ edges and neighborhood size below $k$. One final path invocation, without collection, reaches a first dormant holder from each active source, or every destination when none is dormant.

\emph{Shared selection charge.} All successful inventory selections in both completion stages secure at least $h$ sources from the same entry population; their total is at most $\lfloor a/h\rfloor$. Final empty collection tests cost $R$, not $B$, and total $\Oh(R\log(k+1))=\Oh(B)$. The stopping probe and key search also cost $\Oh(B)$. Data cost is $\Oh(k^2\Lambda\log(k+1))\le\Oh(k^2g)$. Consequently,
\[
 \Oh(B+Ba/h+k^2g)=\Oh(B+(Ba)^{2/3}g^{1/3}).
\]
This is the third branch. Our balanced branch's internal threshold uses only $a,n,N$. All control uses the physical graph; connectivity of the active induced graph is unnecessary.
\end{proof}

\subsection{Interpreting the completion menu}
The completion menu compares three valid procedures: direct path selection uses the remaining population, spreading uses both population and interference, and second pruning trades collection work against the selector scale. Their upper budgets are computable from common parameters, so choosing the minimum requires no structural guess.

Suppressing logarithms, direct completion costs $a^2$, repeated spreading costs $K\sqrt{na}$, and the balanced branch costs $(na)^{2/3}$, in addition to the common control term. Thus $a\le\sqrt n$ already permits nearly linear completion through the direct branch, irrespective of the original degrees. For larger populations the balanced branch improves the dependence on $a$; a sufficiently small contender bound may make spreading cheaper still. This interpretation concerns the handoff population and does not replace $n$ in physical-network control costs.

The logarithmic improvement in the abstract follows from the same balance. The interleaved path selector of~\cite{grx} has length $\Oh(k^2\Lambda\log^2(k+1))$, whereas the permutation selector of~\cite{kuschner} has length $\Oh(k^2\Lambda)$. Allowing $\Oh(\log(a+1))$ doubling iterations, the two path-work budgets at scale $h$ are bounded, respectively, by
\[
 h^2\Lambda\log^3(a+1)\quad\text{and}\quad
 h^2\Lambda\log(a+1).
\]
In either case the shared selection budget is $Ba/h$. Balancing this selection budget with $h^2g$ gives $(Ba)^{2/3}g^{1/3}$. Replacing $g=\Lambda\log^3(a+1)$ by $g=\Lambda\log(a+1)$ therefore reduces this balanced budget by $\log^{2/3}(a+1)$. Both upper budgets use the same control primitive. The factor compares these constructions; the additive $B$ or another completion branch can dominate the final running time.

The source charge is shared across second pruning and all collections: every success consumes a disjoint subset of the same entry population. Empty collections end after their first frozen OR. Thus $\mathcal J(a)$ contains neither a separate degree charge nor a $B\log k$ term.

The third completion branch needs no input bound on residual degree: $K=a$ is always valid. Thus any common secured ledger with $a$ active sources admits quasi-gossiping in $\Oh(B+\mathcal E(a))$, where $\mathcal E(a)=\min\{a^2\Lambda,\mathcal J(a)\}$. This standalone consequence is $\Ot(n+\min\{a^2,(na)^{2/3}\})$. Exact replay must still include the execution that established that ledger; a completion suffix alone is insufficient.

\section{Adaptive Pruning and Path Completion}\label{sec:adaptivepath}
We now apply the completion menu to the adaptive reception-score procedure of Section~\ref{sec:reception}. Its common ledger and contender cap determine a computable completion budget. The controller stops when elapsed pruning time reaches that budget. To analyze this handoff, we compare it with every prefix of one reference execution, whose residual populations remain well defined even after the actual controller stops.

The virtual reference trace uses the adaptive protocol's padded iterations and reset invocation clocks, but continues until $a=0$, followed by an implicit zero-score search. Each positive score secures a source, so our reference trace has $J\le n$ broadcasts. For each $D$, let $a_D$ be the active count at the trace's first below-$D$ search. These are analytical parameters, not inputs. Thresholds share prefixes because the key search is unconditional and snapshots precede the search. For $0\le j\le J$, write $a_j,e_j$ for the reference state after $j$ broadcasts, with $e_0=n$ and $e_j=\min\{e_{j-1},b_j+1\}$.

The handoff compares increasing elapsed time with a decreasing completion budget. An early stop pays for completion through the stopping test; a later stop inherits a budget no larger than at the comparison prefix. This monotonicity permits every reference prefix, including an unreached one, to bound the actual execution. A favorable cover or survivor population can therefore improve the analysis without becoming a protocol input.

\begin{theorem}[Adaptive survivor bound]\label{thm:pathadaptive}
Only common $n,N$ are needed for deterministic gossiping in
\begin{equation}\label{eq:survivorbound}
 \Ot\bigl(\min_{1\le D\le n}\{n(1+\rho_D)+W_D\}\bigr),
\end{equation}
where
\begin{equation}\label{eq:survivormenu}
 W_D=\min\{a_D^2,\ D\sqrt{na_D},\ (na_D)^{2/3}\}.
\end{equation}
The padded implementation has iteration cost $U=\Theta(B)$ and monotone completion budgets $C_j=C(a_j,e_j)$ defined below. Our protocol's total time satisfies
\[
 T\le4\min_j\max\{jU,C_j\}+2U.
\]
Our protocol has no threshold-search slowdown. In particular, $a_D=\Oh(\sqrt{nL}\,\ell)$ and $\rho_D=\Oh(\ell)$ for some $D$ imply $\Oh(nL^2\ell^2)$ rounds. The general $\Ot(n^{4/3})$ guarantee is retained.
\end{theorem}
\begin{proof}
\emph{Common completion budget.} For a sufficiently large common constant $c$, set
\begin{equation}\label{eq:pathbudget}
 C(a,e)=\begin{cases}
 \lceil c[B+\mathcal C(a,\min\{a,2e\})]\rceil,&a>0,\\
 0,&a=0.
 \end{cases}
\end{equation}
The contender argument lies in the domain of Lemma~\ref{lem:survivors}. Both $a$ and $\min(a,2e)$ are nondecreasing in $a,e$, so $C$ is nondecreasing in those variables. A fully secured state needs no completion.

\emph{Executable handoff.} Initialize $F=\varnothing$, elapsed pruning budget $P=0$, and $e=n$. Strong connectivity ensures a positive maximum score whenever $A\ne\varnothing$. Record from the original state. Each iteration probes with the common cap $e$ and performs the complete $\MaxSelect$ operation, padded to $U=\Theta(B)$ rounds. All vertices learn the maximum score $b$. After the winner's broadcast, set
\[
 a=n-|F|,\quad P\leftarrow P+U,\quad e\leftarrow\min\{e,b+1\}.
\]
Stop if $a=0$ or $P\ge C(a,e)$. Otherwise repeat. Before a broadcast, the certificate gives $d_A(v)\le2b$; afterward active degrees only decrease. Hence $d_A(v)<2e$ persists, and closed contender sets have size at most $\min(a,2e)$. Hand off to the cheapest completion branch, ending the outer stopping rule. The completion branch's fixed-threshold pruning is included in $C(a,e)$. Replay the entire first pass, including both stages, by Lemma~\ref{lem:replay}.

\emph{Comparison with a reference prefix.} Let $C_j=C(a_j,e_j)$ along the reference trace, including $j=0$. These budgets are nonincreasing. Fix $0\le j\le J$. If our protocol stops at $i\le j$, completion costs at most $iU$, or is unnecessary when $a_i=0$. Including replay gives $T\le4jU$. If our protocol stops at $i>j$, iteration $i$ began with $(i-1)U<C_{i-1}\le C_j$. Completion costs at most $C_i\le C_j$, giving $T<4C_j+2U$. Minimizing over $j$ proves the prefix inequality. This comparison concerns upper budgets on one greedy trace.

\emph{Profile specialization.} Fix a degree comparator $D$. The prefix before the first below-$D$ search has at most $\rho_D$ broadcasts. If $a_D=0$, the reference ends. Otherwise that search has a positive score $b<D$ and precedes an unconditional broadcast. The ledger update leaves $a_j\le a_D$, while the cap update installs $e_j\le D$, with $j\le\rho_D+1$. This single broadcast explains the extra iteration in the bound; our protocol never searches over candidate thresholds. Monotonicity and the prefix inequality give
\[
 \Oh(U(1+\rho_D)+C(a_D,D)).
\]
This proves~\eqref{eq:survivorbound} simultaneously for every $D$, using the actual $a_D$.

The stated nearly linear specialization uses the first branch: $a_D^2\Lambda=\Oh(nL^2\ell^2)$ and $(1+\rho_D)B=\Oh(nL^2\ell^2)$. For the general guarantee take $D=n$: then $\rho_n=0$, $a_n=n$, and the third branch is $\Ot(n^{4/3})$. Exact replay changes only a constant factor.
\end{proof}

For example, a reference boundary with $j=3$ and $C_3=5U$ gives $T\le12U$ after an earlier handoff, or $T<22U$ after a later one. Both satisfy the prefix inequality. The controller need not predict that boundary or discover the witness that makes its completion budget small.

\section{Uniform Path Covers and Pruning Composition}\label{sec:pathcover}
A path cover describes how many sources can reach useful inventory holders through routes with bounded interference. Our approach first establishes those inventory inclusions and then uses maximum-inventory selection to exploit the cover, without reconstructing the routes. Fix a common secured entry ledger $F_0$, let $C=V\setminus F_0$, and write $a_0=|C|>0$. Let $\mathcal R_k^C(v)$ contain the sources $x\in C$ with a simple carrier path $P$ to $v$ satisfying $|Y_C(P)|\le k$. Include the zero-edge case $x=v$ when $v\in C$. We may assume $1\le k\le a_0$.

A fractional carrier-path cover has exceptions $X\subset C$, $u=|X|<a_0$, and nonnegative receiver weights satisfying
\begin{equation}\label{eq:pathcoverdef}
 \sum_{v:x\in\mathcal R_k^C(v)}w_v\ge1\quad(x\in C\setminus X),
 \qquad \tau=\sum_vw_v\ge1.
\end{equation}
All physical vertices are eligible receivers. The paths, weights, exceptions, and scale are analytical witnesses, not inputs. For $F_0=\varnothing$ this is an original path cover. For a residual cover, $C$ remains the data-carrier set even when later broadcasts secure some carrier vertices. Control always uses the physical graph.

Once a path pass realizes the cover, weighted averaging guarantees a large inventory without revealing the weights or paths. Fixed carriers preserve those paths as sources become secured. Geometric phases increase the affordable selector scale, avoiding separate simulations for every scale.

A secured internal vertex continues relaying while it belongs to $C$, preserving the cover's witness paths. Population completion instead treats a dormant holder as sufficient for quasi-dissemination. The carrier rule may change at that handoff because the cover argument has ended.

\begin{lemma}[Integer-gain cover count]\label{lem:integercover}
Suppose $\mathcal R_k^C(v)\subseteq H_v$ at each maximum-inventory selection, with a fresh saved snapshot allowed at each invocation. For integer $z\ge1$ and real $\tau\ge1$, define
\begin{equation}\label{eq:integercovercount}
 \Gamma(z,\tau)=\begin{cases}
 1,&\tau=1,\\
 \min\{z,\lceil\tau[1+\ln^+(z/\tau)]\rceil\},&\tau>1,
 \end{cases}
\end{equation}
where $\ln^+x=\max(0,\ln x)$ uses natural logarithms. Starting with $a$ unsecured sources, at most $\Gamma(a-u,\tau)$ selections leave at most $u$ whenever $a>u$.
\end{lemma}
\begin{proof}
For $a>u$, at least $a-u$ unsecured sources are covered. Double counting gives
\[
 a-u\le\sum_{x\in C\setminus(F\cup X)}
              \sum_{v:x\in\mathcal R_k^C(v)}w_v
       \le\sum_vw_v|H_v\setminus F|.
\]
The winner's integer gain is at least $\lceil(a-u)/\tau\rceil$. With $z=(a-u)^+$ this yields
\begin{equation}\label{eq:pathcovercontract}
 z_{\rm next}\le\max\{0,z-\lceil z/\tau\rceil\}.
\end{equation}
For $\tau=1$, one selection suffices. Suppose $\tau>1$. If $z\le\tau$, unit progress needs at most $z\le\lfloor\tau\rfloor$ steps; here \eqref{eq:integercovercount} equals $z$. Otherwise, $z_{\rm next}\le(1-1/\tau)z$ and $1-1/\tau\le e^{-1/\tau}$ show that $h=\lceil\tau\ln(z/\tau)\rceil$ steps leave at most $\tau$. At most $\lfloor\tau\rfloor$ more steps reach zero. Combining this with unit progress gives the sharper count
\[
 \min\{z,\lceil\tau\ln(z/\tau)\rceil+\lfloor\tau\rfloor\}
 \qquad(z>\tau>1).
\]
Since $\lceil x\rceil+\lfloor\tau\rfloor\le\lceil x+\tau\rceil$, the compact bound \eqref{eq:integercovercount} is valid and exceeds this count by at most one. Once $a\le u$, monotonicity preserves that inequality.

The proof requires witness inclusions at each invocation, not equality of successive inventories. Replacing $H_v$ by a larger genuine inventory preserves every inequality. Flags, scores, and saved payloads must nevertheless remain frozen within each key search and designated broadcast.
\end{proof}

For example, two inventories jointly covering nine unsecured sources guarantee a maximum gain of at least five, leaving at most four sources. Subsequent maximum selections leave at most two, then one, then zero. Overlap does not weaken the averaging argument. The integer ceiling preserves these final unit gains, which a fractional contraction alone would not certify.

\begin{theorem}[Uniform path-cover bound]\label{thm:pathcover}
One deterministic branch, using only $n,N$, completes gossiping for every original path cover in
\begin{equation}\label{eq:pathcoverfixed}
 \Oh\bigl(k^2\Lambda+B[1+\Gamma(n-u,\tau)]+\mathcal E(u)\bigr).
\end{equation}
Thus our protocol simultaneously attains
\begin{equation}\label{eq:pathcoverbound}
 \Ot\bigl(k^2+n\tau+\min\{u^2,(nu)^{2/3}\}\bigr),
\end{equation}
without a path-scale pool. From a secured ledger $F_0$, the same construction completes quasi-gossiping in additional first-pass time
\[
 \mathcal T=\Oh\bigl(k^2\Lambda+B[1+\Gamma(a_0-u,\tau)]+\mathcal E(u)\bigr)
\]
for any carrier-path cover as above. Here $\mathcal E(0)=0$. Complete-history replay is required for full gossiping.
\end{theorem}
\begin{proof}
\emph{A single geometric execution.}
Pad the path-selector length to $s(k)=c_s k^2\Lambda$ with an integer common constant $c_s$ large enough for Lemma~\ref{lem:path}. Choose integer $U=\Theta(B)\ge s(1)$ bounding a full maximum-key selection and inventory broadcast. Use the scale list $1,2,4,\ldots$ through $a_0$, capping the scale list's last element at $a_0$. Phase $i=0,1,\ldots$ has budget $T_i=2^iU$. First run $S_{k_i}$ on the fixed carriers $C$, where $k_i$ is the largest list element with $s(k_i)\le T_i$. Then perform up to $2^i$ maximum-inventory iterations, each padded to $U$.

At each iteration entry, save fresh $H_v=K_v$ and score $b_v=|H_v\setminus F|$. Search unconditionally for the largest score and broadcast that winner's saved inventory. Credit only this designated snapshot to the common ledger; actual packets retain all genuine knowledge. All vertices, including those outside $C$, may win. A source outside $F$ knows its own rumor, so a positive score exists until $F=V$. Every iteration secures a new source. The counterfactual reference execution that never invokes early completion therefore terminates after at most $a_0$ iterations. This inventory-score trace is separate from the degree-score reference execution.

\emph{Observable stopping rule and phase overshoot.}
Let $P$ be all elapsed path and selection time since entry. For a large common constant $c$, use the monotone completion budget
\[
 C_H(a)=\begin{cases}\lceil c[B+\mathcal E(a)]\rceil,&a>0,\\0,&a=0.\end{cases}
\]
After every selection, stop if $a=n-|F|=0$ or $P\ge C_H(a)$. At a nonterminal stop, use the population-only completion of Lemma~\ref{lem:survivors}. Population completion may reduce the active set between frozen data invocations; the fixed carrier convention ends at this handoff.

At reference selection boundaries write $P_j,a_j,C_j=C_H(a_j)$, with $P_0=0$. The sequence $C_j$ is nonincreasing. Within one phase the next boundary adds exactly $U$. Before phase $i\ge1$, the preceding full phase paid $2^{i-1}U=T_i/2$ in selections. Hence phase $i$'s data pass adds at most $T_i\le2P_{j-1}$. At the first boundary the pass adds at most $U$. Uniformly,
\[
 P_j\le3P_{j-1}+2U.
\]
Fix any reference comparator $j$, and let $h$ be the actual handoff boundary. Here $\mathcal T$ counts all first-pass time since entry, including the completion suffix. If $h\le j$, a nonterminal handoff has completion budget at most $P_h$, giving $\mathcal T\le2P_h\le2P_j$. When $a_h=0$, no completion is needed and the same bound holds. If $h>j$, the prior stopping test gives $P_{h-1}<C_{h-1}\le C_j$. The phase bound and monotonicity then give $P_h\le3C_j+2U$ and $C_h\le C_j$. Consequently,
\[
 2\mathcal T\le8\min_{0\le j\le J}\max\{P_j,C_j\}+4U.
\]
This accounts for whole path passes at phase transitions. Applying the preceding constant-iteration comparison directly would omit these jumps.

\emph{Amortized witness cost.}
Round witness $k$ to the first list scale $k'\ge k$, so $k'\le2k$. Let $p$ be the first phase with $k_p\ge k'$. Then $T_p=\Oh(U+k^2\Lambda)$. That phase's completed pass establishes $\mathcal R_k^C(v)\subseteq K_v$ at every receiver. Knowledge subsequently grows, so every later saved snapshot preserves these inclusions. Fixing $C$ is essential: removing newly secured carriers could remove an internal path sender. Inventory growth causes no such problem.

Let $t=\Gamma(a_0-u,\tau)$. Lemma~\ref{lem:integercover} ensures $a\le u$ within $t$ selections after this pass. If $a\le u$ is reached before phase $p$, that comparator costs below $4T_p$, including terminal cases. Otherwise let $q$ contain the first post-pass selection leaving $a\le u$; at most $t$ such selections are needed. A full phase costs at most $2T_i$, so all phases through $q$ cost less than $4T_q$. If $q=p$, this is $\Oh(U+k^2\Lambda)$. If $q>p$, completing phase $q-1$ required fewer than $t$ of these selections, so $2^{q-1}<t$ and $T_q<2Ut$. Thus some reference comparator satisfies
\[
 P_j=\Oh(B+k^2\Lambda+Bt),\qquad C_j=\Oh(B+\mathcal E(u)).
\]
The phase-aware prefix comparison proves the additional first-pass bound. For $F_0=\varnothing$, replay of the entire path, control, and completion history costs another $\mathcal T$ and proves~\eqref{eq:pathcoverfixed}. Since $\Gamma(z,\tau)=\Oh(1+\tau\log(z+1))$ and $B=\Ot(n)$, the theorem's soft form is~\eqref{eq:pathcoverbound}. Every branch is correct without a witness. A constant share of the full portfolio preserves the portfolio's earlier guarantees.
\end{proof}

A phase-boundary example illustrates the overshoot calculation. If $T_i=8U$, the next path pass and first selection cost at most $9U$. The preceding phase has already paid $4U$ in selections, so $P_{j-1}\ge4U$ and $9U\le2P_{j-1}+U$. Thus elapsed work pays for the entire phase jump.

Pruning can now supply the entry ledger for the same geometric construction. The next result keeps the returned carrier set fixed throughout those phases.

\begin{corollary}[Amplified pruning followed by a residual path cover]\label{cor:prunedpathcover}
For fixed public $D$, run $\Prune(D)$, keep the pruning procedure's actual returned ledger $F_D$, and pin $C=A_D=V\setminus F_D$ throughout the geometric phases. If $a_D>0$ and this carrier set admits a cover $(k,\tau,u)$, full gossiping costs
\[
 \Oh\bigl(k^2\Lambda+B[1+\rho_D+\Gamma(a_D-u,\tau)]+\mathcal E(u)\bigr).
\]
If $a_D=0$, pruning already completes gossiping. With only $n,N$ as inputs, a private pool simultaneously gives
\[
 \Ot\bigl(n(1+\rho_D+\tau)+k^2+\min\{u^2,(nu)^{2/3}\}\bigr)
\]
for every cover at an actual dyadic-threshold state.
\end{corollary}
\begin{proof}
Pruning costs $T_0=\Oh((1+\rho_D)B)$ by Lemma~\ref{lem:charge}. Run the uniform construction on the pruning procedure's returned ledger, giving quasi-gossiping after $T_0+\mathcal T$ rounds from the original rumors. Exact replay must include this entire history, not just the new suffix. Lemma~\ref{lem:replay} gives full gossiping in $2(T_0+\mathcal T)$.

Use private branches for $D=1,2,4,\ldots$, capped at $n$, at $\Oh(\log n)$ slowdown. Each branch keeps its own carrier list and recording from the original state. A post-replay frozen OR of $|K_v|<n$ verifies completion. No witness is needed to run or verify a branch. The minimum is over these actual states: rounding an arbitrary integer $D$ can change $F_D$, available carrier paths, and the cover. We assert no constant-factor rounding rule for residual witnesses.
\end{proof}

For $\tau=1,u=0$, the original-cover bound is $\Oh(k^2\Lambda+B)$ without knowing $k$. More generally, $\Gamma\le a_0-u$ and integer gains sharpen $\tau\log n$ to $\Oh(\tau[1+\ln^+((a_0-u)/\tau)])$. For $z=n$, $\tau=n/\log n$, the count becomes $\Oh(n\log\log n/\log n)$ instead of $\Oh(n)$. Geometric execution also removes the logarithmic overhead of independent scale simulations. These improve explicit budgets without claiming a better general polynomial exponent. Amplified pruning can lower the path scale; fixing residual carriers preserves the witness paths.

At linear density, sufficient residual conditions are $\rho_D+\tau=\Ot(n^{1/5})$, $k=\Ot(n^{3/5})$, and $u=\Ot(n^{4/5})$ at an actual dyadic state. A sparsity-only guarantee for these residual witnesses is not proved here. A direct cover need not have small path scale, so the separate direct-reception branch remains useful.

\section{Reception Covers and Structural Improvements}\label{sec:cover}
Path covers require a data pass to assemble useful inventories. Direct reception covers instead use one-hop adjacency, so the degree probes already provide the necessary evidence. A weighted cover lower-bounds the largest reception score; the winner's saved inventory then secures a corresponding set of active origins. Iterating this gain contracts the unsecured population.

At reference prefix $j$, a residual reception cover consists of exceptions $U\subset A_j$ and nonnegative receiver weights $w_v$ such that
\begin{equation}\label{eq:coverdef}
 \sum_{v:(x,v)\in E}w_v\ge1\quad(x\in A_j\setminus U),\quad \tau=\sum_vw_v.
\end{equation}
Write $u=|U|<a_j$; necessarily $\tau\ge1$. These are analytical witnesses, not inputs. Unit weights on receivers cover the union of their in-neighborhoods. The averaging step below is the greedy-cover argument also used in~\cite[Section 6]{wu}; here our averaging step acts on certified neighbor counts.

\begin{theorem}[Residual-cover contraction]\label{thm:fractional}
Write $\mathcal E(a)=\min\{a^2\Lambda,\mathcal J(a)\}$. At any reference prefix $j$ with $a_j>u$, consider a residual cover of weight $\tau$ and $u$ exceptions. For $1\le r\le a_j-u$, put
\[
 \lambda=1-(1-1/Q)/\tau,\qquad
 t=\left\lceil\frac{\log((a_j-u)/r)}{\log(1/\lambda)}\right\rceil.
\]
The same adaptive protocol takes
\begin{equation}\label{eq:fractional}
 \Oh\bigl(B(1+j+t)+\mathcal E(u+r)\bigr)
\end{equation}
rounds. Thus our protocol's bound is $\Ot(n(j+\tau)+(nu)^{2/3})$, simultaneously for all such witnesses. Original covers have $j=0$; only $n,N$ are inputs.
\end{theorem}
\begin{proof}
\emph{From a cover to a certified gain.} Later active sets are subsets of $A_j$, so the same weights still cover every remaining nonexceptional source. At population $a>u$, at least $a-u$ sources are covered. Double counting gives
\[
 a-u\le\sum_{x\in A}\sum_{v:(x,v)\in E}w_v
       =\sum_vw_vd_A(v)\le\tau\max_vd_A(v).
\]
The maximum score is at least $(1-1/Q)\max_vd_A(v)$. The winning receiver's saved inventory contains that many distinct active neighbor origins, giving
\begin{equation}\label{eq:covercontraction}
 a_{\rm next}-u\le\lambda(a-u).
\end{equation}
Once $a\le u$, monotonicity preserves the inequality $a\le u$. After $t$ further reached iterations the population is at most $u+r$.

\emph{From a reference prefix to the actual stopping time.} If the adaptive protocol stops before prefix $j+t$, our protocol's cost, including completion and replay, is $\Oh(B(1+j+t))$. Otherwise the reference boundary there has completion budget at most $\Oh(B+\mathcal E(u+r))$. Apply the prefix inequality of Theorem~\ref{thm:pathadaptive}; if the reference ends earlier, gossiping is already complete. This also covers $t=0$. Finally $-\ln\lambda\ge(1-1/Q)/\tau$ gives $t=\Oh(1+\tau\log n)$. Set $r=1$ for the soft bound; the additive $n$ absorbs the constant offset when $u=0$.
\end{proof}

\begin{corollary}[Successive residual covers]\label{cor:stagedcovers}
Let $J$ be the number of positive reference broadcasts. Start at $j_0=0$ and take $q\ge1$ analytical stages with $j_i<J$ for $i<q$. Stage $i$ has a reception cover on $A_{j_i}$ of weight $\tau_i$, exceptions $u_i<a_{j_i}$, and target $1\le r_i\le a_{j_i}-u_i$. Define $t_i$ by Theorem~\ref{thm:fractional} and $j_{i+1}=\min\{J,j_i+t_i\}$. Put $s=0$ if $j_q=J$, and $s=u_{q-1}+r_{q-1}$ otherwise. The same protocol takes
\[
 \Oh\!\left(B\left(1+\sum_{i<q}t_i\right)+\mathcal E(s)\right).
\]
In soft notation this is $\Ot(n\sum_{i<q}\tau_i+(ns)^{2/3})$. No stage, cover, or target is supplied to the protocol.
\end{corollary}
\begin{proof}
At stage $i$, the cover remains valid for every later active subset of $A_{j_i}$. The contraction proof leaves at most $u_i+r_i$ sources after $t_i$ further iterations. Induction justifies each reached boundary $j_{i+1}$ and leaves at most $s$ sources at the final one. Earlier secured sources never reenter any covered population.

The cumulative prefix contains at most $\sum_i t_i$ iterations. Apply the prefix inequality there, bounding the prefix's completion budget by $\Oh(B+\mathcal E(s))$. If the adaptive execution stops before that boundary, our protocol's completion is already paid by the controller's elapsed budget. If the reference ends first, all original sources are secured. Both cases fit the same expression. Finally $t_i=\Oh(1+\tau_i\log n)$ and $\tau_i\ge1$ give the soft form. The stages are an analysis of one unchanged execution, so no stage-search slowdown is hidden.
\end{proof}

\begin{corollary}[Minimum and quantile out-degree]\label{cor:outdegree}
Let $u_d=|\{x:\deg^+(x)<d\}|$ for $1\le d\le n-1$. With only $n,N$, gossiping takes
\[
 \Ot\bigl(\min_d\{n^2/d+(nu_d)^{2/3}\}\bigr).
\]
In particular, minimum out-degree $\delta^+$ gives $\Ot(n^2/\delta^+)$, with $\delta^+$ unknown.
\end{corollary}
\begin{proof}
Put weight $1/d$ on every receiver and except sources of out-degree below $d$. This gives $\tau=n/d$ and $u=u_d$; apply Theorem~\ref{thm:fractional}. If $u_d=n$, use the general branch, which fits the displayed bound. Take $d=\delta^+$ for the last claim. Neither $d$ nor the out-degree quantile is an input.
\end{proof}
For directed $\Delta$-regular graphs, this improves both the $\Ot(n\sqrt\Delta)$ benchmark of~\cite{wu} and the general $\Ot(n^{4/3})$ bound when $\Delta=n^{\gamma+o(1)}$ with $\gamma>2/3$. The exponent becomes $2-\gamma$.

\subsection{A nonregular family for the quantile bound}\label{sec:nonregularquantile}
Corollary~\ref{cor:outdegree} applies even when minimum out-degree is one. If $u_d=n^{\alpha+o(1)}$ sources have out-degree below $d=n^{\beta+o(1)}$, its two terms give exponent
\begin{equation}\label{eq:quantileexponent}
 \max\{2-\beta,\;2(1+\alpha)/3\}.
\end{equation}
The first term pays for the reception-cover weight $n/d$; the second completes the exceptional population. A few low-degree sources therefore need not determine the full cost.

To realize this distinction, choose $1\le u$ and $2\le d<n-u$. On $n-u$ core vertices, form a directed circulant by connecting each vertex to its next $d$ vertices in cyclic order. At one core vertex $c$, attach the directed cycle $c,p_1,\ldots,p_u,c$. The circulant contains a directed Hamilton cycle on the core, and the attached cycle makes the whole graph strongly connected. Every core vertex has out-degree at least $d$, while each peripheral vertex has out-degree one. Hence $u_d=u$, $\delta^+=1$, and
\[
 m=(n-u)d+u+1.
\]
For sufficiently large $n$, the choice $u=d=\lfloor n^{4/5}\rfloor$ is admissible and gives $m=\Theta(n^{9/5})$. Equation~\eqref{eq:quantileexponent} yields $\Ot(n^{6/5})$, saving $2/15$ in the exponent against the general $\Ot(n^{4/3})$ upper bound. The minimum-out-degree specialization alone gives only $\Ot(n^2)$. No matching lower bound is asserted.

Neither the core nor $d,u$ is supplied to the protocol. The same reception-score execution exploits the quantile witness through its reference-prefix comparison. The next result uses a different witness: one receiver with few missing incoming arcs.

\subsection{A near-universal receiver}
\begin{corollary}[A near-universal receiver]\label{cor:cover}
Suppose receiver $z$ has incoming arcs from all but at most $u$ other vertices. With all other degrees unrestricted, gossiping takes
\begin{equation}\label{eq:hubbound}
 \Oh\bigl(nL^2\ell^2+(u+1)^2L\bigr)
\end{equation}
and also
\begin{equation}\label{eq:coverbound}
 \Ot\bigl(n+(nu)^{2/3}\bigr).
\end{equation}
In particular, $u=\Oh(\sqrt{nL}\,\ell)$ gives $\Oh(nL^2\ell^2)$. No knowledge of $z,u$ is required.
\end{corollary}
\begin{proof}
Put $s=u+1$. At population $a>s$, the maximum active in-degree is at least $a-s$. Since $q(\kappa)\ge nQ/a$, a winning inventory leaves
\[
 a_{\rm next}\le s+(a-s)/q(\kappa)\le s+a^2/(nQ).
\]
If $s\ge a^2/(nQ)$, the next population is at most $2s$. Otherwise $a_{\rm next}\le2a^2/(nQ)$. Until reaching $2s$, the normalized population $2a/(nQ)$ squares at each step and starts at most $2/Q\le1/2$. Thus $\Oh(\ell)$ iterations suffice. The prefix comparison gives $\Oh(B\ell+\mathcal E(2s))$. The direct path-selection branch of our population-only completion proves~\eqref{eq:hubbound}; the balanced branch of our completion procedure proves~\eqref{eq:coverbound}, including $u=0$. The stopping rule needs neither $z$ nor $u$.

A finite upper count is $2+\lceil\log_2(1+\log_2(nQ))\rceil$: otherwise repeated squaring would leave fewer than one active source while still exceeding $2s$. This comparison uses the reference trace; the actual controller can hand off sooner.
\end{proof}

\subsection{Why a cover gives a different amortization}
The weighted degree-tail charge counts how often individual receivers can win at a fixed score threshold. A reception cover instead lower-bounds the best score using the current population. Our reception-cover strategy can yield fast pruning even when many receivers have large original degrees. Equation~\eqref{eq:covercontraction} uses only the gain from the selected neighbors; any additional rumors in the saved inventory improve the gain. The contraction remains valid when the winner is dormant, when exceptions have already been secured, or when the maximum-scoring receiver changes between iterations.

For a fractional example, suppose every nonexceptional source has arcs into at least two of three receivers. Weight $1/2$ on each gives $\tau=3/2$. Some receiver then has at least $2(a-u)/3$ active in-neighbors, so the maximum certified score is at least $(1-1/Q)2(a-u)/3$. Overlapping neighborhoods preserve the bound. The receiver identities and weights need not be known.

Fix the cover's entry prefix $j$. After $h$ additional broadcasts, the population is $a_{j+h}$. Each step leaves at most a $\lambda$ fraction of the previous positive excess over $u$; monotonicity preserves zero excess. Induction gives
\[
 (a_{j+h}-u)^+\le\lambda^h(a_j-u).
\]
No exact neighborhood count or local estimate of the cover is required. The lower bound on the saved maximum score follows directly from Lemma~\ref{lem:probe}. The score's upper bound is not needed for this argument.

\subsection{Interpretation and scope of reception covers}
The cover inequality counts arcs from each source to the weighted receivers. Reversing the arcs in~\eqref{eq:coverdef} would invalidate the double-counting step. Consequently Corollary~\ref{cor:outdegree} uses minimum \emph{out}-degree; large minimum in-degree alone does not supply that witness. Equality of the two degrees is unnecessary.

For one near-universal receiver, adaptive amplification makes the residual recurrence quadratic until the exception count dominates. This improves the control term from $\Oh(nL^3)$ to $\Oh(B\ell)=\Oh(nL^2\ell^2)$. A cover of arbitrary fixed weight greater than one need not have this quadratic contraction. Distributed constructions can instead benefit from the truncated degree profile when only a few receivers have degree at least two.

At the original state $j=0$, let $\kappa_u(G)$ be the minimum fractional reception-cover weight after allowing at most $u$ exceptional sources, for $0\le u<n$. The minimum exists: weights above one can be truncated to one without violating a source constraint, leaving a compact feasible domain for each of finitely many exception sets. Strong connectivity ensures feasibility by assigning weight one to every receiver. The invariant is nonincreasing in $u$ and at least one. Theorem~\ref{thm:fractional} therefore gives the envelope
\begin{equation}\label{eq:coverenvelope}
 \Ot\bigl(\min_{0\le u<n}\{n\kappa_u(G)+(nu)^{2/3}\}\bigr).
\end{equation}
The envelope holds on one execution because the stopping proof compares it with every cover. No exception-set search is performed. The second term pays for completion of the exceptional population, which remains available in the physical control network.

In the distributed construction below, each $z_i$ has only its cycle successor as an out-neighbor, so $\delta^+=1$ gives a quadratic uniform-weight bound. A reception cover of fixed weight $h$ nevertheless leaves only $h$ exceptions and gives nearly linear communication. Dense graphs can also have useful fractional covers without a small set of near-universal receivers. Arbitrary weights and exceptions capture both possibilities.

\begin{corollary}[A cover criterion for the density benchmark]\label{cor:coverdensity}
The $\Ot((mn)^{3/5})$ benchmark holds for arbitrary polynomial identifiers whenever some reference prefix $j$ has a residual reception cover with
\[
 j+\tau=\Ot((mn)^{3/5}/n),\quad u=\Ot(m^{9/10}n^{-1/10}).
\]
For $m=\Theta(n)$, the sufficient conditions are $j+\tau=\Ot(n^{1/5})$ and $u=\Ot(n^{4/5})$. No global cover or degree-tail hypothesis is needed.
\end{corollary}
\begin{proof}
The terms $n(j+\tau)$ and $(nu)^{2/3}$ in Theorem~\ref{thm:fractional} both fit within the stated benchmark. The cover is a witness for analysis and need not be supplied or reconstructed.
\end{proof}

\begin{proposition}[Joint coverage without a near-universal receiver]\label{prop:jointcover}
For every fixed integer $h\ge2$, there are sparse strongly connected graphs with linear diameter, no receiver adjacent from all but $o(n)$ other vertices, and polynomial-identifier gossiping time $\Oh(nL^2\ell)$.
\end{proposition}
\begin{proof}
Take a directed Hamilton cycle in the order
\[
 z_1,P_1,z_2,P_2,\ldots,z_h,P_h,z_1,
\]
where each $P_i$ is a directed block of $\Theta(n/h)$ peripheral vertices. Add the arc $(p,z_i)$ for every $p\in P_i$, omitting duplicates. There are between $n$ and $2n$ edges. Unit weights on $\{z_1,\ldots,z_h\}$ cover every peripheral source, leaving at most $h$ exceptions. Only the $h$ receivers $z_i$ have in-degree at least two, so $t_2\le h$ and $\chi_2\le h$. Theorem~\ref{thm:clearadaptive} with $D=2$ gives $\Oh(nL^{3/2}\ell^{1/2}+(1+h)B)=\Oh(nL^2\ell)$ for fixed $h$. This argument uses the clearing branch directly and does not require a survivor estimate from another execution.

A peripheral receiver has in-degree one. A receiver $z_i$ has in-degree at most $|P_i|+1=n/h+\Oh(1)$ if the blocks are balanced, so every receiver misses $\Omega(n)$ vertices. From $z_i$ to the last vertex of $P_i$, every added arc returns to $z_i$; no added arc skips forward in $P_i$. The distance is therefore $|P_i|=\Theta(n)$ for fixed $h$. The cycle ensures strong connectivity and the distance also proves an $\Omega(n)$ gossiping lower bound. This family shows why joint coverage is strictly more general than a single near-universal receiver.
\end{proof}

The older reception-score analysis also applies to the distributed construction. A below-$2$ search gives $d_A(v)\le1$, leaving at most one peripheral source per block and at most $h$ hub sources. Thus $a_2\le2h$ and $\rho_2\le hH_\star$. Theorem~\ref{thm:pathadaptive} gives $\Oh(B(1+hH_\star)+\mathcal E(2h))=\Oh(nL^2\ell^2)$ for fixed $h$. The clearing proof above removes one factor $\ell$.

\begin{proposition}[Sparse graphs with linear diameter]\label{prop:core}
There are strongly connected graphs with $m=\Theta(n)$, maximum in-degree and diameter $\Theta(n)$, and $\Theta(\sqrt n)$ receivers of in-degree $\Theta(\sqrt n)$, on which the protocol takes $\Oh(nL^2\ell)$ rounds for arbitrary polynomial identifiers. This is optimal in the polynomial exponent and violates condition~\eqref{eq:concentration}.
\end{proposition}
\begin{proof}
Let $r=\lfloor\sqrt n\rfloor$ and $f=n-r-1$. Use the directed Hamilton cycle
\[
 z,p_1,\ldots,p_f,c_1,\ldots,c_r,z.
\]
Add $(p_i,z)$ for every $i$ and all missing arcs between distinct vertices of $\{c_1,\ldots,c_r\}$. The graph is simple and strongly connected. The graph's edge count lies between $n$ and $2n+r^2$, while $z$ has in-degree $f+1=\Theta(n)$. Every core vertex has in-degree at most $r$, and every $p_i$ has in-degree one. Among receivers of degree at least two, the closed neighborhoods are the core, the core plus $p_f$, and $\NB(z)\cup\{z\}$. Hence $t_2\le3$. Theorem~\ref{thm:clearadaptive} with $D=2$ gives $\Oh(B)=\Oh(nL^2\ell)$; the finitely many small cases are absorbed by the constant.

The distance from $z$ to $p_f$ is $f=\Theta(n)$: an additional arc from a peripheral vertex goes back to $z$, and core arcs have both endpoints beyond $p_f$. Thus $\Omega(n)$ rounds are necessary for the original rumor of $z$ to reach $p_f$. Finally the average degree is bounded, whereas the $r$ core receivers have growing degree. They violate~\eqref{eq:concentration} for every fixed $C_0,C_1$ when $n$ is sufficiently large.
\end{proof}

The sparse core also retains the earlier receiver-based guarantee: $z$ misses at most $r=\lfloor\sqrt n\rfloor$ other vertices, so Corollary~\ref{cor:cover} gives $\Oh(nL^2\ell^2)$. Clearing its three high-degree classes improves this to $\Oh(nL^2\ell)$.

\begin{corollary}[A two-parameter family]\label{cor:parameter}
For integers $1\le s\le u\le n/2$, there are strongly connected graphs with $m=\Theta(n+su)$ and linear diameter and maximum in-degree that admit
\begin{equation}\label{eq:parameter}
 \Ot(n+[n(u-s)]^{2/3})
\end{equation}
rounds for polynomial identifiers. At $u=\lfloor n^{7/10}\rfloor$, $s=\lfloor n^{1/2}\rfloor$, this is $\Ot(n^{17/15})$, saving $14/75$ against the density exponent $33/25$. If $s=u$, the time is $\Oh(nL^2\ell)$ with an $\Omega(n)$ lower bound. Taking $s=u=\lfloor n^{7/10}\rfloor$ gives $m=\Theta(n^{7/5})$, saving $1/3$ against the general exponent $4/3$.
\end{corollary}
\begin{proof}
Use the cycle $z,p_1,\ldots,p_f,c_1,\ldots,c_u,z$, where $f=n-u-1$. Add all $(p_i,z)$ and the arcs from $s$ selected core vertices to all other core vertices, omitting duplicates. The cycle ensures strong connectivity and $m=\Theta(n+su)$.

Put unit reception weights on $z,c_u$. Each peripheral source reaches $z$; each selected core source other than $c_u$ reaches $c_u$; and $c_u$ reaches $z$. The uncovered sources are contained in $\{z\}$ and the unselected core, so there are at most $u-s+1$. Theorem~\ref{thm:fractional} with $\tau=2$ proves~\eqref{eq:parameter}; the additive $n$ absorbs the constant exception offset. 

When $s=u$, there are at most three distinct closed in-neighborhoods among receivers of degree at least two: the core, the core plus $p_f$, and $\NB(z)\cup\{z\}$. Thus $\chi_2\le t_2\le3$. Theorem~\ref{thm:clearadaptive} with $D=2$ gives $\Oh(B)=\Oh(nL^2\ell)$, independently of the number of core vertices. The small cases $u\le2$ satisfy the same estimate by the same bound on $t_2$.

Vertex $z$ has in-degree $f+1=\Theta(n)$ and distance $f=\Theta(n)$ to $p_f$: added peripheral arcs return to $z$, and core arcs cannot skip that path. This proves the diameter and $\Omega(n)$ lower bound. Substitution gives density exponent $33/25$ and time exponent $17/15$ in the first case; $s=u=n^{7/10+o(1)}$ gives density exponent $7/5$ and time exponent one. The savings compare upper bounds.
\end{proof}

\subsection{Exact accounting for the two-parameter family}\label{sec:familyaccounting}
A single receiver missing $u=n^{\alpha+o(1)}$ sources gives exponent $\max\{1,2(1+\alpha)/3\}$. In Corollary~\ref{cor:parameter}, the second weighted receiver reduces the exceptions to $u-s+1$, while the exact arc count below determines the density exponent.

Write $C=\{c_1,\ldots,c_u\}$ and let $S\subseteq C$ be the $s$ selected core sources. Counting the added arcs exactly gives
\begin{equation}\label{eq:familyexactedges}
 m=n+f+s(u-1)-|S\setminus\{c_u\}|.
\end{equation}
The subtraction removes precisely the selected-source core arcs already on the Hamilton cycle; all peripheral arcs into $z$ are new. Thus $m=\Theta(n+su)$, including $u=1$.

Unit weights at $z,c_u$ cover every peripheral source, every selected core source, and $c_u$ itself. A valid exceptional set is
\[
 U_*=(\{z\}\cup(C\setminus S))\setminus\{c_u\},
 \qquad e=|U_*|\le u-s+1.
\]
At a reference boundary, the maximum score is at least $(1-1/Q)(a_j-e)/2$ whenever $a_j>e$. Consequently,
\[
 (a_{j+1}-e)^+\le\left(1-\frac{1-1/Q}{2}\right)(a_j-e)^+.
\]
The factor is at most $3/4$, so $\lceil\log_{4/3}n\rceil$ selections suffice to leave at most $e+1$ sources. The prefix comparison gives $\Oh(B\log n+\mathcal E(e+1))$, even if the actual handoff occurs earlier. Dormant weighted receivers remain eligible, so the witness persists.

When $s=u$, $U_*=\{z\}$. The earlier reception-score analysis also gives $\Oh(BH_\star+\mathcal E(3))=\Oh(nL^2\ell^2)$: at most three high-degree neighborhood classes qualify before the first below-$2$ search. At that search, $b\le1$ and integrality give $d_A(v)\le1$ everywhere. The two unit-weight receivers cover all but $z$, so $a-1\le d_A(z)+d_A(c_u)\le2$ and $a\le3$. Clearing strengthens this to $\Oh(B)$ by securing each of those classes once. Both bounds apply throughout $m=\Theta(n+u^2)$, up to quadratic density; the improved statement in Corollary~\ref{cor:parameter} uses the clearing execution. The distance $f=\Theta(n)$ establishes the linear polynomial lower bound, without matching the logarithmic factors.

\emph{Why the complete core is inexpensive.} When $s=u$, many high-degree receivers share one closed neighborhood. Clearing an elected receiver secures every source in that neighborhood, so no other receiver in its class can later justify a positive-score phase. At most three high-degree classes are charged. The bounded-degree completion then costs $\Oh(B)$ at comparator $D=2$. The protocol uses certified scores and common ledgers without discovering the core or the classes. This explains the improved logarithmic budget on the dense family.

The survivor bound also transfers the sparse density benchmark to arbitrary polynomial identifiers. The sufficient conditions are that some $D$ has $\rho_D=\Ot(n^{1/5})$ and $a_D=\Ot(n^{4/5})$. Substitution in~\eqref{eq:survivorbound} proves $\Ot(n^{6/5})$ without an exact frontier.

\section{Condensing on a Frozen Active Set}\label{sec:transfer}
The condensing analysis of~\cite{wu} uses round robin for packet accumulation and for prescribed delivery across a frontier. These uses impose different guarantees. Our packet certificate handles accumulation by making sufficiently large pending knowledge a collection witness. Exact one-hop coverage is retained for a particular frontier packet, with its cost controlled by residual interference after pruning.

\subsection{Proof strategy and delivery requirements}
The path-based completion procedures above use ordered delivery. We next develop an alternative based on layered condensing: an unresolved packet either reaches a higher receiver class or reaches a holder whose own original rumor is globally known. Progress through the finitely many classes establishes quasi-gossiping. Pruning fixes the data sender set at entry, while all physical vertices remain available for control.

Condensing requires a second ledger. The original-source ledger $F_0$ stays fixed and determines the data senders; the phase ledger records owners whose entry packets have been collected. Collecting an owner's packet makes its own original rumor global; any target retained by that owner is then quasi-disseminated, even if learned after the snapshot was created. Retention and complete-history replay preserve this progress.

The analysis separates time near collections from the remaining time. Collections consume pending packet owners and are countable. Away from collections, reaching a packet threshold would force a flush at the next frozen test. The certificate must therefore deliver the whole neighboring pending union, giving a latency bound for a specified target.

A frontier into an unprocessed class has no corresponding collection implication. With residual active in-degree bounded by $\kappa$, exact one-hop delivery costs $\Ot(\min\{N,\kappa^2\})$. This separate requirement explains why the density transfer has conditions. The proof next specifies the class budgets, aligned frames, and finite stopping horizon.

The indices play different roles. A source phase $s$ follows packets \emph{owned} by vertices of $V_s$; a receiver channel $j$ controls interference on edges \emph{entering} $V_j$. For processed high channels, $j\le s$. Their packet certificates use the source threshold $\lceil\beta_s\rceil$, while $\beta_j$ bounds the channel period because $\beta_s\le\beta_j$. Eligibility counts packets by owner class at every holder, regardless of the holder's own class. Keeping these roles separate explains both the latency sum and the common collection bound.

\emph{Proof outline.} We first align complete data invocations between packet-ledger updates. A threshold outcome then forces a countable collection, so invocations away from collections deliver the full neighboring packet union. A finite phase horizon leaves enough subsequent tests to act on any persistent witness. Finally, an exact schedule crosses the unprocessed frontier, where the collection implication is unavailable.

\begin{lemma}[Frozen-active-set condensing]\label{lem:active}
Suppose a common $F_0$ consists of globally known original sources. Freeze $A=V\setminus F_0$ as the data sender set. Let $\kappa\ge1$ be a common upper bound on $d_A(v)$ for every receiver, and let $M\ge\max(n,\sum_v d_A(v))$ be common. Put $H_\kappa=n+\sqrt{M\kappa}$. A finite procedure reaches quasi-gossiping in
\begin{equation}\label{eq:activecost}
 \Ot(\min\{N,\kappa^2\}+(MH_\kappa)^{3/5})
\end{equation}
rounds. Control uses the full physical graph. The data subgraph may be disconnected and may have zero in-degrees. If $d_A(v)=0$ for every $v$, quasi-gossiping already holds.
\end{lemma}
\begin{proof}
We specify our condensing procedure and give the finite-horizon analysis. Throughout the procedure, only $A$ transmits in data slots. All vertices may transmit in assigned control slots. Every rumor retained by a dormant holder is quasi-disseminated, because the holder's own source lies in $F_0$. Accordingly, a temporal reception path of a rumor that is still not quasi-disseminated cannot pass through a dormant vertex. The original-source ledger $F_0$ and the data sender set are never changed by the phase packet collections below.

\subsubsection*{Degree classes, including zero in-degrees}
Run Lemma~\ref{lem:probe} on $A$ with closed-contender bound $\min(n,\kappa+1)$ and $Q=2$, and put
\[
 \widehat d_v=\max\{1,\min(n-1,2b_v)\}.
\]
If all scores are zero, every active source would have to have no out-neighbor, contradicting strong connectivity unless $A$ is empty. Thus the claimed zero-degree case is already quasi. Otherwise, frozen score tests find the common maximum $K=\max_v\widehat d_v$. Then
\[
 d_A(v)\le\widehat d_v\le\max\{1,2d_A(v)\},\qquad K\le2\kappa.
\]
Use layers $V_j=\{v:k_j/2<\widehat d_v\le k_j\}$, where $k_j=2^j$ and $0\le j\le\lceil\log K\rceil$. Each vertex knows its own class and includes the class index in phase packet metadata. An estimate of one can represent a true zero degree; estimates in larger classes satisfy $d_A(v)>k_j/4$.

Clamp $M$ to $\min(M,nK)$ and keep the notation $M$. The result is still at least $n$ and still bounds $\sum_vd_A(v)$, because $K\ge1$ and $d_A(v)\le K$. Define common weights
\begin{equation}\label{eq:layerweights}
 w_j=4M/k_j.
\end{equation}
For $j\ge1$, summing $d_A(v)>k_j/4$ gives $|V_j|\le w_j$. For $j=0$, the separate inequality $|V_0|\le n\le M$ proves the same bound, including all zero-degree vertices. This base-layer treatment is necessary; a positive-degree counting inequality alone would not bound their number.

The mixed scale $H=n+\sqrt{MK}$ meets two frame requirements. The $n$ term covers physical control and the base class, including zero-degree receivers. The square-root term lets the largest high-channel period fit between ledger updates. Thus frame feasibility depends on both the physical network size and residual interference; class-size estimates will separately bound the number of collections.

\subsubsection*{Frames and thresholds}
Write $H=n+\sqrt{MK}$. For a sufficiently large common constant $c_0$, let $T_0$ be the first power of two at least $c_0H\log^3N$. Thus
\begin{equation}\label{eq:T0journal}
 T_0=\Theta(H\log^3N),\qquad x=(T_0M)^{1/5}.
\end{equation}
Indices with $k_j<x$ are low; the remaining ones are high. Set
\begin{equation}\label{eq:transferparams}
 \alpha_j=(T_0M/k_j^3)^{1/2},\qquad
 \beta_j=(T_0M/k_j^2)^{1/3}
\end{equation}
for low and high classes, respectively. Eligibility uses $t_j=\lceil \alpha_j\rceil$ in a low class and $t_j=\lceil \beta_j\rceil$ in a high class. These thresholds are at least one. A threshold larger than $n$ never creates an eligible witness; whenever that threshold would be passed to the packet certificate, use a strong sender selector instead.

These choices balance travel time against collection time. For a tentative class threshold $t$, a bounded path uses at most $t$ incoming edges in that class. A low channel has latency proportional to $k_j^2$; the high-channel certificate instead has latency proportional to $k_jt$, with logarithms suppressed. Since the class has at most $4M/k_j$ owners and each collection consumes $t$, the contributions to balance are
\[
 \begin{array}{ll}
 \text{low:}&tk_j^2+T_0M/(k_jt),\\
 \text{high:}&k_jt^2+T_0M/(k_jt).
 \end{array}
\]
Equating the two terms in each row gives~\eqref{eq:transferparams}. The formulas therefore pay for both movement and interruptions by collection; optimizing only the data latency would not justify the phase budget.

A low receiver-class channel repeats a strong selector for $\min(n,k_j+1)$ closed contenders. Let the channel's power-of-two padded period be
\[
 \lambda_j=\Oh(k_j^2\log N).
\]
In a high source phase $s$, each processed high receiver-class channel $j\le s$ repeats Lemma~\ref{lem:union} with threshold $\lceil \beta_s\rceil$, amplification two, and closed contender bound $\min(n,k_j+1)$. The certificate's packet sets consist only of entry-pending packets owned by $V_s$. Since $\beta_s\le \beta_j$, pad this period to
\begin{equation}\label{eq:gammajournal}
 \gamma_j=\Oh(k_j\beta_j\log^{3/2}N).
\end{equation}
If $\beta_s>n$, the strong-selector shortcut has no larger order: $k_j\le2n$ and $\beta_j\ge \beta_s>n$ imply $k_j^2=\Oh(k_j\beta_j)$.

Every data period fits in $T_0/4$ and divides $T_0$. For a high period, let $c_\gamma$ bound the constant in its cubed cost, including power-of-two padding. Since $k_j\le2K$,
\[
 \frac{\gamma_j^3}{T_0^3}
 \le\frac{2c_\gamma MK\log^{9/2}N}{T_0^2}
 \le\frac{2c_\gamma}{c_0^2(\log N)^{3/2}}.
\]
Here $T_0^2\ge c_0^2MK\log^6N$ and $\log N\ge1$. Choosing $c_0^2\ge128c_\gamma$ gives $\gamma_j\le T_0/4$; the excess logarithmic power supplies slack. For low periods, if $K^2\le n$, then $\lambda_j=\Oh(n\log N)$. Otherwise $M\le nK<K^3$, so $(n+\sqrt{MK})^3\ge M^2$. The restriction $k_j<x$ then gives
\[
 \lambda_j\le\Oh((T_0M)^{2/5}\log N)\le T_0/4.
\]
The last inequality follows from $T_0^3=\Omega(M^2\log^9N)$ and a sufficiently large frame constant. Rounding periods to powers of two changes constants only. The control selection budget $B$ also fits inside a frame.

\subsubsection*{Phase state and the actual protocol}
Initialization creates a phase-zero packet at every vertex, containing its entry inventory, owner identifier, and computed layer. Only packets carrying this metadata enter eligibility counts. Run all low channels and use only low-class eligibility, with the high sums in the budgets below empty. Use the same frame horizon and append the same exact one-hop schedule as in a high phase. The initialization goal is quasi-dissemination or a holder in an active high layer. Subsequent high source layers are processed in increasing order. At each phase entry every vertex creates a fresh packet $\pi_u$ containing its full inventory, own rumor, and owner layer; all original-rumor values are retained.

Let $\mathcal L$ be a common phase ledger, initially empty. A packet is pending exactly when its owner is not in $\mathcal L$. Let $P_v^j$ be the pending phase packets known at $v$ whose owners belong to $V_j$. A vertex is eligible if
\begin{equation}\label{eq:phaseelig}
 (\exists\hbox{ low }j:|P_v^j|\ge\lceil \alpha_j\rceil)
 \ \vee\
 (\exists\hbox{ high }j\le s:|P_v^j|\ge\lceil \beta_j\rceil).
\end{equation}
One frozen selection is run per logical frame. At frame entry, save eligibility flags and the complete phase-packet inventories. If the selection is nonempty, broadcast the winner's saved packet inventory and add the saved packets' owner identifiers to $\mathcal L$ at the frame boundary. Credit only that saved inventory. The phase ledger is separate from $F_0$; a phase collection never changes $A$.

The following steps specify one high source phase.
\begin{enumerate}
\item Create the entry snapshots and reset $\mathcal L$. Start all low channels and all high receiver-class channels $j\le s$ at a common period boundary.
\item At every data invocation, freeze each sender's pending $V_s$ projection. Packet gating uses that projection for the entire invocation. Every actual payload appends all phase packets and all original rumors currently known.
\item In parallel logical time, execute one frozen control selection per frame using~\eqref{eq:phaseelig}. Delay the selection's ledger commitment to the common frame boundary.
\item Run for the common phase budget specified below. Then execute one exact one-hop data schedule, again with sender set $A$ and full payloads.
\item Retain all original-rumor values, discard obsolete phase objects, and continue to the next high source class.
\end{enumerate}
There are $\Oh(\log n)$ logical channels. All frame and latency budgets below use logical time. Serialize channel slots in a common order on one half-duplex radio, adding a logarithmic factor to the physical running time. Freeze logical decisions at round boundaries; full payloads may include information received in earlier serialized slots. Control ignores extra fields, and frozen data gates preserve the certificate and charging arguments.

There are consequently two kinds of freezing. Entry-pending projections determine which senders participate during one data invocation. Saved eligibility and inventories determine which packet owners a control selection can credit. Payload contents may continue to grow in both operations, but that growth changes neither decision. The packet ledger is committed only at aligned boundaries, preserving the pending-universe premise of every certificate.

\subsubsection*{Counting collections and paid frames}
A successful selection has a saved witness in some processed source class $j$. The selection's broadcast removes at least $t_j$ previously pending owner identifiers of that class. Charge the selection to one such class, for example the first qualifying one. Since the universe of that class has size $|V_j|$, the number of successes is at most
\begin{equation}\label{eq:Fsjournal}
 F_s=\sum_{k_j<x}\frac{w_j}{\lceil \alpha_j\rceil}
       +\sum_{x\le k_j\le k_s}\frac{w_j}{\lceil \beta_j\rceil}.
\end{equation}
This is a common upper bound; actual layer sizes need not be learned. Our collection bound remains valid when the saved inventory also broadcasts packets of other classes, because those extra removals can only reduce later pending sets.

Call a flush frame and its two preceding frames paid. Also mark the final two phase frames paid. All other frames are unpaid, so there are at most $3F_s+2$ paid frames. These labels are assigned retrospectively in the analysis; our protocol never branches on paid status. The final-two-frame convention keeps each unpaid frame's subsequent frozen tests inside the phase.

In an unpaid frame no receiver can acquire $\lceil \beta_s\rceil$ pending packets from $V_s$. Indeed, there is no flush in that frame or in its next two frames. Those packets would persist to the next frame-entry test, making a vertex eligible and causing a flush in the next frame, a contradiction. Pending here refers to the current phase ledger, not to actual global knowledge inferred from some unrelated transmission.

Consider an invocation of Lemma~\ref{lem:union} contained in an unpaid frame. The invocation's designated universe is the set of $V_s$ owners not in $\mathcal L$ at entry. Already flushed packets are excluded from both activation and the threshold count. No commitment occurs inside the invocation because periods divide frames. The threshold alternative would contradict the preceding paragraph. The invocation must therefore deliver \emph{every} neighboring entry-pending $V_s$ packet.

Deleting paid frames is only an accounting device: the protocol executes every frame and retains all knowledge. A target can wait through paid time for an unpaid invocation. Aligned periods bound this waiting and prevent a certificate from straddling a ledger commitment, preserving its frozen pending universe.

\subsubsection*{Unpaid edge latency}
Suppose a target packet from $V_s$ is at $u$ and $(u,v)$ enters a processed high receiver layer $j$. Until the target becomes quasi-disseminated or is flushed, the target remains in the sender's pending projection. If a complete invocation remains in the current unpaid frame, the next full one finishes within $2\gamma_j$ unpaid rounds after an arbitrary arrival time. If none remains, fewer than $\gamma_j$ unpaid rounds remain before the frame ends. Skip paid frames; the next unpaid frame begins at a complete invocation boundary and contributes at most another $\gamma_j$. Arrival during a paid frame is even simpler, since no time is charged before the next unpaid frame. Thus the unpaid delay is at most $2\gamma_j$.

For a low receiver class, strong selectors give an ordinary delay at most $2\lambda_j$, and hence the same upper bound in unpaid time. These bounds concern a specified pending target. The all-or-threshold interface and frame alignment supply the required delivery guarantee.

To compose the edge bounds, fix a processed-layer path starting at the target's owner at phase entry. Induct on its prefixes: the target reaches each prefix endpoint within the sum of the corresponding unpaid edge latencies, unless it has already become quasi-disseminated or reached a higher layer. The base case is its presence at the owner. At each step, retention and the unchanged pending status permit the next edge bound; flushing the target already resolves it. The horizon below supplies sufficient unpaid time for every bounded path.

Call a processed-layer path bounded if the path has at most $\lceil \alpha_j\rceil$ vertices in each low class and at most $\lceil \beta_j\rceil$ in each processed high class. Summing the appropriate incoming-edge delays gives
\begin{equation}\label{eq:Asjournal}
 A_s=2\sum_{k_j<x}\lceil \alpha_j\rceil\lambda_j
       +2\sum_{x\le k_j\le k_s}\lceil \beta_j\rceil\gamma_j.
\end{equation}
Overcounting the path's starting vertex is harmless. Round $A_s$ up to a frame boundary. With $C_s=\lceil F_s\rceil$, reserve
\begin{equation}\label{eq:Tsjournal}
 T_s=A_s+4T_0(C_s+1).
\end{equation}
The first $A_s$ unpaid rounds end at a frame boundary $\tau$ with
\begin{equation}\label{eq:taubudget}
 \tau\le A_s+(3C_s+2)T_0,\qquad
 T_s-\tau\ge(C_s+2)T_0.
\end{equation}
Such a $\tau$ exists: even after deleting all possible paid frames, more than $A_s$ unpaid rounds remain. The reserved tail contains at least $C_s+2$ complete frames, hence more frozen tests than the total permitted number of successful selections.

Before $\tau$, the unpaid travel budget suffices for every bounded route; afterward, persistent eligibility would force too many collections. This analytical boundary is not a protocol input. The procedure always runs to its prescribed common phase boundary.

For example, $A_s=3T_0$ and $C_s=2$ give a fifteen-frame phase. At most eight frames are paid, so the first three unpaid frames end by frame eleven, leaving at least four complete tests. Persistent eligibility would force four successes, contradicting the bound of two. The travel budget alone would not ensure this response time.

\subsubsection*{Path witnesses and the frontier}
Fix $\pi_v$ with $v\in V_s$. Quasi-dissemination and retention by a higher layer are monotone. We may assume neither objective holds before the appended exact schedule; otherwise the phase has already succeeded for this target. At $\tau$, choose an unreceived vertex $w$ of minimum graph distance from $v$. A predecessor $u$ on a shortest path has received the target. Let $P$ be a loop-erased temporal reception path to $u$. Deleted excursions can be replaced by waiting. Every vertex packet on $P$ traveled with the target along the remaining suffix, because payloads append all known phase packets.

\emph{A processed frontier forces collection.} The path has neither a dormant vertex nor a higher-layer vertex, either of which would resolve the target. If $w$ is processed, $P$ followed by $(u,w)$ cannot be bounded: summing the edge latencies along that fixed vertex sequence, starting at $v$ at phase entry, would deliver the target by $\tau$. This comparison does not restart the latency clock at the observed arrival at $u$. Some class therefore exceeds its vertex threshold. Removing the final vertex leaves at least the threshold number of distinct owner packets on $P$, all held at $u$.

None of these witness packets can have been flushed. If an owner's packet was flushed before it received the target, that reception would already make the target quasi-disseminated. A later flush has the same consequence by retention. Thus every witness owner remains pending, and $u$ stays eligible until a resolving flush occurs. Since $\tau$ is a frame boundary, each of the next $C_s+2$ complete frames begins with a frozen test witnessing this eligibility. Every such selection must succeed, exceeding the total cap $C_s\ge F_s$. The processed-frontier case therefore forces the phase objective before the exact schedule.

\emph{An unprocessed frontier needs exact delivery.} If $w$ is unprocessed, $u$ retains the target until phase end. Append the shorter of $N$ singleton slots and a strong selector for $\min(n,K+1)$ closed contenders, using only $A$ as data senders. The cost is $\Ot(\min\{N,\kappa^2\})$. The predecessor $u$ is active, and every receiver has at most $K$ active in-neighbors. The exact schedule therefore delivers this particular target to $w$. A dormant $w$ makes the target quasi; an active $w$ is in a higher layer. A threshold-only guarantee could instead deliver other packets and leave this target unresolved.

For initialization, the same argument has only low processed classes. A phase-zero packet already held at an active high layer meets the initialization goal; every other unresolved packet is forced into such a layer by the appended exact schedule. If there is no high layer, initialization already establishes quasi-gossiping. At entry to each high phase $s$, the induction invariant is
\begin{equation}\label{eq:condensinginvariant}
 \begin{gathered}
 r\text{ not quasi-disseminated}\quad\Longrightarrow\\
 r\in K_v\text{ for some }v\in A\cap\bigcup_{j=s}^{\lceil\log K\rceil}V_j.
 \end{gathered}
\end{equation}
A holder above $s$ already meets the next goal. Otherwise a holder in $V_s$ places the rumor in its entry snapshot, and the phase makes that snapshot quasi or moves it above $s$. At the last high phase there is no higher layer, so every remaining rumor becomes quasi-disseminated. The target need not originate at its current holder; the entry snapshot preserves the whole inventory. This induction uses no connectivity of the active induced graph.

\subsubsection*{Cost calculation}
The cutoff balances opposite degree dependencies: low-class strong-selector costs increase with class scale, while balanced high-class costs decrease. Their meeting point explains the fifth-root cutoff and three-fifths exponent. The calculation below bounds each contribution before summing over classes.

For a low class, ignoring only logarithmic factors,
\begin{align*}
 \alpha_j\lambda_j&=\Ot((T_0M)^{1/2}k_j^{1/2}),\\
 T_0w_j/\alpha_j&=\Ot((T_0M)^{1/2}k_j^{1/2}).
\end{align*}
Since $k_j<x=(T_0M)^{1/5}$, both are at most $\Ot((T_0M)^{3/5})$. For a high class,
\begin{align*}
 \beta_j\gamma_j&=\Ot((T_0M)^{2/3}k_j^{-1/3}),\\
 T_0w_j/\beta_j&=\Ot((T_0M)^{2/3}k_j^{-1/3}),
\end{align*}
and $k_j\ge x$ gives the same upper bound. The ceilings change these estimates by constants, since thresholds are at least one. Rounding $A_s$ and adding frames costs $\Oh(T_0)$ per phase. With $M\ge n$ and $K\le n-1$, we have $n+\sqrt{MK}=\Oh(M)$, so these additive frames are absorbed by $(T_0M)^{3/5}$.

Equating these two scale expressions, with logarithms suppressed, gives $k_j^{5/6}=(T_0M)^{1/6}$, hence $k_j=(T_0M)^{1/5}$. Substitution yields $(T_0M)^{3/5}$. Our low/high split therefore equalizes the largest polynomial contributions. The separate frame-fit calculation remains necessary to keep complete data invocations between ledger updates.

There are $\Oh(\log n)$ layers, phases, and logical channels. The exact end schedules contribute $\Ot(\min\{N,\kappa^2\})$ in total. Since $K\le2\kappa$, the other costs are at most the second term of~\eqref{eq:activecost}. Degree estimation and common score tests fit below that bound. This proves the lemma.
\end{proof}

The frozen-set lemma needs common bounds on interference and active arcs. Pruning supplies the first; verified geometric guesses supply the second. Each trial retains its own original state and complete recording, so an underestimated guess can lose progress but cannot falsely certify completion.

\subsection{Interpreting the path witness}
An unresolved target creates an overrepresented owner class at a holder when a bounded processed route fails. The reserved tests force collection. For example, $x$ may create and forward $\pi_x$ before learning target $\pi_v$. A later collection of $\pi_x$ need not contain $\pi_v$, but it disseminates $x$'s own rumor while $x$ retains the target. This establishes quasi-dissemination; complete-history replay supplies full delivery.

The proof uses two clocks. The protocol executes every frame chronologically, while travel is charged only in unpaid time. A flush in frame $f$ may charge $f-2,f-1,f$; overlapping charges reduce the total. Retention preserves targets during these interruptions. Bounded unpaid travel, bounded paid time, and the reserved tests together give the chronological horizon. Neither retrospective frame labels nor $\tau$ is a protocol input. At an unprocessed frontier, the appended exact schedule supplies the prescribed delivery that the threshold argument cannot ensure.

\section{Pruned Density Transfer}
We now combine pruning with the frozen-active-set completion. A threshold $D$ creates three costs: securing enough sources to reduce contention, providing exact residual frontier coverage, and running the remaining condensing phases. Theorem~\ref{thm:density} displays all three costs so that a favorable degree profile can pay for the first without concealing the second. The proof then removes the common edge-budget and threshold inputs by verified trials and a dyadic pool. Those trials compare complete executions, each with its own state and full replay history.

\begin{theorem}[Pruned density transfer]\label{thm:density}
For $1\le D\le n$, put $E_D=\min\{N,D^2\}$, $H_D=n+\sqrt{mD}$, and $P_D=n(1+\rho_D)+E_D$. With only common $n,N$, gossiping takes
\begin{equation}\label{eq:transfer}
 \Ot(\min\{n^{4/3},\min_D[P_D+(mH_D)^{3/5}]\}).
\end{equation}
In particular, the $\Ot((mn)^{3/5})$ bound of~\cite{wu} holds with polynomial identifiers if some $D$ satisfies
\begin{equation}\label{eq:transferconditions}
 n\rho_D+E_D=\Ot((mn)^{3/5}),\qquad mD=\Ot(n^2).
\end{equation}
For $m=\Theta(n)$, some $D=\Ot(n^{3/5})$ with $h_D=\Ot(n^{1/5})$ suffices. The original maximum in-degree is unrestricted.
\end{theorem}
\begin{proof}
\emph{Fixed threshold and a valid edge budget.} Fix $D$ and record from the original rumors. Run $\Prune(D)$, obtaining ledger $F_0$, pruning cost $\Oh((1+\rho_D)B)$, and residual active in-degrees below $2D$ (Lemma~\ref{lem:charge}). Freeze the active set and apply Lemma~\ref{lem:active} with $\kappa=\min(n-1,2D)$ and an edge upper bound $M\ge m$. The physical control network remains strongly connected; common frame budgets give a common quasi-gossiping boundary.

Replay the entire first pass, including pruning, to obtain gossiping by Lemma~\ref{lem:replay}. The exact coverage cost is at most $\Ot(E_D)$ because $\kappa=\Oh(D)$. For a valid guess $M\le2m$, the remaining condensing cost is $\Ot((mH_D)^{3/5})$. The pruning term is $\Ot(n(1+\rho_D))$, giving the fixed-$D$ expression.

\emph{Verified edge-budget guesses.} Use guesses $n,2n,4n,\ldots$, capped at $n(n-1)$. Each trial starts from genuine initial rumors and includes pruning and complete replay. Its frame and phase horizons are computed from the guess and are finite even if that guess underestimates the active edge count. The protocol executes one bounded selection per frame and stops at the prescribed horizon; a false collection-count estimate affects the progress proof, not termination of the trial.

After replay, freeze the flags $|K_v|<n$ and compute their OR; accept only a false answer. An underestimated guess cannot falsely certify all $n$ genuine source values. The first valid guess is at most $2m$ and completes correctly. The algebraic bounds on prescribed horizons increase with the edge-budget parameter up to constant and logarithmic factors, so each earlier trial fits within the valid trial's upper budget. There are $\Oh(\log n)$ trials, absorbing repeated pruning and verification in~\eqref{eq:transfer}.

\emph{Unknown pruning threshold.} Use dyadic $D$ values and cap the final one at $n$. For any analytical comparator $D$, the first covering value $D'$ satisfies $D\le D'\le2D$. Monotonicity gives $\rho_{D'}\le\rho_D$, while $E_{D'}\le4E_D$ and $H_{D'}\le\sqrt2H_D$, so rounding loses only a constant. Interleave these $\Oh(\log n)$ private branches and the general branch~\cite{kuschner}. The logarithmic slowdown is allowed by $\Ot$. Each simulation keeps separate ledgers, knowledge, clocks, and verification flags; our protocol stops at the first common certified completion.

\emph{Density and sparse specializations.} The second hypothesis in~\eqref{eq:transferconditions}, $mD=\Ot(n^2)$, gives $H_D=n+\sqrt{mD}=\Ot(n)$. Its first hypothesis bounds $n\rho_D+E_D$ by the density budget; the additive $n$ also fits because $m\ge n$. For sparse graphs, $D=\Ot(n^{3/5})$ gives $E_D=\Ot(n^{6/5})$ and $H_D=\Ot(n)$. Equation~\eqref{eq:coarsetail} gives $\rho_D=\Ot(h_D)$, so $h_D=\Ot(n^{1/5})$ makes the pruning cost $\Ot(n^{6/5})$. No original-degree cap enters this argument.
\end{proof}

\begin{corollary}[Uniform identifier-range extension]\label{cor:uniformspan}
Every graph in the model admits gossiping in
\begin{equation}\label{eq:uniformspan}
 \Ot(\min\{n^{4/3},N+(mn)^{3/5}\})
\end{equation}
rounds, using only common $n,N$. Consequently, the density benchmark $\Ot(\min\{n^{4/3},(mn)^{3/5}\})$ holds whenever $N=\Ot((mn)^{3/5})$, and in particular for every graph when $N=\Ot(n^{6/5})$.
\end{corollary}
\begin{proof}
For $m\le n^{11/9}$, choose $D=\lceil n^2/m\rceil\le n$. Then
\[
 n\rho_D\le n^2/D\le m\le(mn)^{3/5},\quad mD\le n^2+m,
\]
while $E_D\le N$. Thus Theorem~\ref{thm:density} costs $\Ot(N+(mn)^{3/5})$. For $m\ge n^{11/9}$, the general $n^{4/3}$ branch is no slower than $(mn)^{3/5}$. Taking the faster branch proves~\eqref{eq:uniformspan}. The consequences follow from $m\ge n$; threshold adaptation removes $m$ as an input.
\end{proof}

The additive identifier penalty separates the proven extension from the remaining gap. If $N=n^{\alpha+o(1)}$ and $m=n^{c+o(1)}$, the unconditional exponent is at most $\min\{4/3,\max\{\alpha,3(c+1)/5\}\}$. For $6/5<\alpha<4/3$, the density benchmark is therefore retained whenever $c\ge5\alpha/3-1$; the general branch covers $c\ge11/9$. The profile theorem can do better when $N$ is larger. No assertion of the density benchmark for every polynomial range follows from this corollary.

\emph{Why the range is uniform.} Strong connectivity gives $m\ge n$, so the smallest density budget is $n^{6/5}$. This absorbs $N=\Ot(n^{6/5})$ on every graph; denser graphs can absorb more. Degree concentration supplies nearly edge-linear gossiping for arbitrary polynomial ranges.

\subsection{Relation to the unpruned transfer}
Choose $D=\Delta+1$ when $\Delta<n-1$, and $D=n$ otherwise. Then $h_D=0$ and $\rho_D=0$. Also $D=\Oh(\Delta)$ because strong connectivity gives $\Delta\ge1$. Thus~\eqref{eq:transfer} contains the unpruned bound
\[
 \Ot\!\left(\min\{N,\Delta^2\}
       +[m(n+\sqrt{m\Delta})]^{3/5}\right).
\]
The new theorem additionally permits expensive original degrees when they can be removed with a small weighted tail charge. Our pruned-density strategy therefore does more than extend the regular case or assume that every receiver already has a bounded degree.

The unrestricted polynomial-identifier density extension remains open. For an unfavorable profile, pruning and exact residual coverage can exceed the density budget at every useful threshold. The theorem exposes that tradeoff and retains the general fallback; it does not preclude a different selector or frontier mechanism.

\subsection{A sparse family with linear maximum in-degree}
Let $n=q^5$ and partition the vertices into $q$ groups of $q^4$ vertices. In every group choose a center and direct an arc from each other group member into that center. Add a directed Hamilton cycle through all vertices, omitting duplicate arcs. Finally add every missing arc into one distinguished center. The graph is simple and strongly connected. This graph has between $n$ and $3n$ arcs, so $m=\Theta(n)$; its distinguished center has in-degree $n-1$.

Every noncenter has in-degree one, and each nondistinguished center has in-degree between $q^4-1$ and $q^4$. At $D=\lceil q^3\rceil$, for sufficiently large $q$, exactly the $q$ centers have degree at least $D$. Consequently $D=\Theta(n^{3/5})$ and $h_D=\Theta(n^{1/5})$, which satisfy the sparse transfer theorem. However, the average degree is bounded, so for every fixed $C_0,C_1$ the $q$ high-degree centers violate~\eqref{eq:concentration} for sufficiently large $q$. The original maximum-degree sufficient condition also fails. The family demonstrates that the pruned-density theorem covers profiles beyond both earlier specializations. The distinguished receiver has in-degree $n-1$, while every other receiver has in-degree at most $q^4$. For large $q$, its amplification-$4$ reception score exceeds $q^4$, so the clearing branch elects it in its first phase. Clearing this universal receiver secures all sources and gives $\Oh(B)=\Oh(nL^2\ell)$ rounds on this family. The profile calculation remains useful to show that the density transfer does not require a small original maximum degree.

\section{Structural Consequences and Scope}\label{sec:structural}
Clearing, reception covers, and density transfer exploit different graph features. We now specialize them to degree concentration, regularity, and degree tails. Each structural assumption supplies an analytical comparator for a protocol using only $n,N$; the resulting bounds identify when large original degrees remain affordable.

\begin{corollary}[Nearly edge-linear gossiping]\label{cor:edge}
Under~\eqref{eq:concentration}, deterministic gossiping takes
\begin{equation}\label{eq:edgepoly}
 \Oh(mL^2\ell)
\end{equation}
rounds for every $N=n^{\Oh(1)}$, using only common $n,N$. No additional bound on the maximum in-degree or on $\Delta/\dbar$ is required.
\end{corollary}
\begin{proof}
Put $D_0=\lceil C_0\dbar\rceil$. Since $\dbar\ge1$, $nD_0=\Oh(m)$. If $D_0\le n$, condition~\eqref{eq:concentration} and~\eqref{eq:clearcharge} give $\chi_{D_0}\le t_{D_0}\le h_{D_0}=\Oh(\dbar)$. Theorem~\ref{thm:clearadaptive} therefore gives
\[
 \Oh\!\left(mL^{3/2}\ell^{1/2}+n(1+\dbar)L^2\ell\right)
 =\Oh(mL^2\ell).
\]
If $D_0>n$, take $D=n$. Then $t_n=\chi_n=0$ and $n^2<C_0m=\Oh(m)$, yielding the same conclusion. The threshold and structural constants are comparators in the analysis, not protocol inputs.
\end{proof}

\subsection{Earlier concentration refinements}\label{sec:legacyconcentration}
For comparison, the reception-score branch under~\eqref{eq:concentration} retains the earlier bound
\begin{equation}\label{eq:legacyconcentration}
 \Oh\!\left(mL^2\bigl[\ell+\min\{\ell^2,\log(1+\Delta/\dbar)\}\bigr]\right).
\end{equation}
Indeed, at $D_0=\lceil C_0\dbar\rceil\le n$, Lemma~\ref{lem:charge} gives
\[
 \rho_{D_0}=\Oh\!\left(\dbar\min\{H_\star,1+\log^+(\Delta/D_0)/\ell\}\right).
\]
Substitution in Theorem~\ref{thm:profile} proves~\eqref{eq:legacyconcentration}; if $D_0>n$, use $D=n$ and $m=\Omega(n^2)$. Exact pruning gives $\Oh(mL^2\ell)$ when $N=\Oh(B)$. For polynomial identifiers, the earlier branch also gives that bound if $\Delta\le\dbar L^a$ for fixed $a\ge0$, or whenever $\Phi_D=\Oh(\dbar)$ at some $D=\Theta(\dbar)$. For example, $h_{Q^jD}\le C\theta^j$ with $C=\Oh(\dbar)$ and fixed $0<\theta<1$ supplies the latter condition.

These comparisons preserve the progression of the guarantees. Fixed-accuracy receiver charging gives $\Oh(mL^3)$; adaptive accuracy gives~\eqref{eq:legacyconcentration} and hence $\Oh(mL^2\ell^2)$. Clearing gives $\Oh(mL^2\ell)$ under concentration alone, removing the additional maximum-degree or summable-tail requirement for that logarithmic budget.

For $m=\Theta(n^c)$, $1\le c<4/3$, condition~\eqref{eq:concentration} saves exponent $\min\{4/3-c,(3-2c)/5\}>0$ against the shared linear-identifier benchmarks~\cite{kuschner,wu}; the saving is $1/9$ at $c=11/9$.

\emph{How the average degree pays for clearing.} Choosing $D$ proportional to $\dbar$ makes the completion factor $nD$ proportional to $m$. Condition~\eqref{eq:concentration} leaves at most $\Oh(\dbar)$ relevant closed-neighborhood classes. Charging each once at cost $\Oh(B)$ gives $\Oh(\dbar B)=\Oh(mL^2\ell)$. A large individual degree is affordable because all broadcasts needed to clear its class fit within that one charge.

An edge count alone gives only $h_D\le m/D$. At a threshold proportional to $\dbar$, this permits $\Theta(n)$ high receivers, which can be much larger than $\dbar$. Our current charging argument therefore needs the concentration condition or another progress witness. This distinction separates the proved nearly edge-linear class from the unrestricted sparse-graph question.

\subsection{Linear diameter and large maximum in-degree}
Fix $1\le c<4/3$, and let $h=\lfloor n^{c-1}\rfloor$, $s=\lfloor n/2\rfloor$, $\tau=n-s$. In a core of $s$ vertices, connect $h$ hubs to every nonhub in both directions. At one nonhub $a$, attach a directed cycle $a,t_1,\ldots,t_\tau,a$ through the remaining vertices. For sufficiently large $n$, $s-h\ge n/4$, and
\[
 m=2h(s-h)+\tau+1=\Theta(n^c).
\]
The diameter and maximum in-degree are $\Theta(n)$. Reaching $t_\tau$ from $a$ requires $\tau$ edges. Only hubs have in-degree above $h+1$. Moreover, $\dbar\ge(h+1)/2$, so $3\dbar>h+1$ and $h\le2\dbar$. Thus~\eqref{eq:concentration} holds with $C_0=3,C_1=2$, realizing the preceding exponent savings on graphs with these large parameters.

\subsection{Regular graphs with polynomial identifiers}
Regularity permits completion without a prescribed frontier. Spreading creates many holders of every rumor through out-ball growth, and maximum-inventory collection converts this multiplicity into global progress. Selectors and control depend logarithmically on the identifier range, with no $N$-slot sweep. Thus the regular-graph extension is complete despite the frontier cost in general condensing.

The second regular-graph bound uses a different consequence of regularity. Uniform weight $1/\Delta$ on every receiver covers every source, since each source has $\Delta$ out-neighbors. The total cover weight is $n/\Delta$, so our reception-cover strategy gives $\Ot(n^2/\Delta)$. Our protocol retains the better of this bound and the spreading bound. The two mechanisms exploit outgoing coverage and balanced-cut expansion, respectively.

For $\Delta=n^{\gamma+o(1)}$, $0\le\gamma\le1$, these mechanisms give exponents $1+\gamma/2$ and $2-\gamma$. The smaller exponent is strictly below $4/3$ for fixed $\gamma\ne2/3$. At $\gamma=2/3$, our stated upper bound matches the general polynomial exponent; no strict improvement at this crossover is asserted.

\begin{theorem}[Regular gossiping with polynomial identifiers]\label{thm:regular}
If every vertex has both in-degree and out-degree $\Delta$, then gossiping takes $\Ot(\min\{n\sqrt\Delta,n^2/\Delta\})$ rounds for $N=n^{\Oh(1)}$, using only common $n,N$. This extends the identifier range of~\cite[Theorem 2]{wu}; $\Delta$ need not be known.
\end{theorem}
\begin{proof}
A $Q=2$ probe on $A=V$ and common score tests give $\Delta\le K\le2\Delta$ without changing $F$. For $K\ge n^{2/3}$ use the general branch. Otherwise run $q=\min\{n-1,\lceil n/K^{3/2}\rceil\}$ strong-selector periods for $\min(n,K+1)$ contenders and freeze $H_v=K_v$. The out-ball expansion $|B_q^+(u)|\ge\min(n,q\Delta/3)$ from~\cite[Lemma 16]{wu} gives $f=\Omega(n/\sqrt\Delta)$ holders per rumor. To verify the expansion, let $L_j$ be a nonempty distance layer and $S$ the radius-$(j-1)$ out-ball. If $|L_j|+|L_{j+1}|\ge\Delta$, the inequality below is immediate. Otherwise at least $|L_j|(\Delta-|L_j|-|L_{j+1}|)>0$ arcs enter $S$ from $L_j$. All arcs leaving $S$ go from $L_{j-1}$ to $L_j$, so there are at most $|L_{j-1}||L_j|$ of them. Equal in- and out-degrees balance this cut. Dividing by $|L_j|>0$ gives
\[
 |L_{j-1}|+|L_j|+|L_{j+1}|\ge\Delta.
\]
The argument uses directed cut balance and assumes no reverse arcs.

For radii $1,2,3$, the first out-ball already contains $\Delta+1$ vertices. For larger radii before full coverage, summing the last three layers and inducting in steps of three gives $|B_q^+(u)|\ge q\Delta/3$.

While $a=n-|F|>0$, $\sum_v|H_v\setminus F|\ge fa$. Thus $\MaxSelect(|H_v\setminus F|)$ removes at least $fa/n$ sources per call. Fixed witnesses and actual snapshots containing them give $\Oh(\sqrt\Delta\log n)$ calls. Spreading costs $\Oh((n\sqrt K+K^2)\log N)=\Oh(n\sqrt\Delta\log N)$; collection costs $\Oh(B\sqrt\Delta\log n)$. On any graph the procedure terminates correctly since $v\in H_v$, so the portfolio needs no regularity test.

The additional $\Ot(n^2/\Delta)$ bound follows from Corollary~\ref{cor:outdegree}. Interleave the two valid protocols at constant slowdown and accept the first common certified completion. Neither $\Delta$ nor regularity is an input. For $\Delta=n^{\gamma+o(1)}$, $\gamma>2/3$, the added branch has exponent $2-\gamma<4/3$ and improves the $1+\gamma/2$ regular benchmark of~\cite{wu}.
\end{proof}

\subsection{Polynomial degree tails}
\begin{corollary}[Polynomial-tail portfolio]\label{cor:polytail}
If $h_k\le Cn/k^p$ for fixed $p>1,C$, gossiping takes $\Ot(n^{\beta(p)})$, where
\begin{equation}\label{eq:polytail}
 \beta(p)=\begin{cases}
 4/(p+2),&1<p<4/3,\\
 6/5,&4/3\le p\le4,\\
 1+1/(p+1),&p>4.
 \end{cases}
\end{equation}
At $p=5$, the profile branch gives the explicit logarithms in Table~\ref{tab:bounds}.
\end{corollary}
\begin{proof}
The level sum gives $\Phi_D\le CnD^{-p}/(1-Q^{-p})$ and $m=\Theta(n)$. The profile bound balances at $D=\lceil(n\sqrt{L\ell})^{1/(p+1)}\rceil$, with exponent $1+1/(p+1)$ and logarithmic factors $L^{3/2+1/(2(p+1))}\ell^{1/2+1/(2(p+1))}$. Theorem~\ref{thm:density} with $D=\lceil n^{2/(p+2)}\rceil$ gives $\Ot(n^{4/(p+2)}+n^{6/5})$. Taking the faster branch gives~\eqref{eq:polytail}; no tail parameter is an input.
\end{proof}

\section{Further Tail Consequences and Limits}
We next derive the exponent regimes in the preceding portfolio and extend the calculations to tails that hold only above a threshold. We then examine residual-edge sensitivity and the limits of the certificates. This order separates consequences of the proved bounds from the additional mechanism still needed for an unrestricted density extension.

\subsection{Deriving the polynomial-tail portfolio}
Under the global tail assumption $h_k\le Cn/k^p$ with fixed $p>1$, summing degree tails gives
\[
 m=\sum_{k\ge1}h_k\le Cn\sum_{k\ge1}k^{-p}=\Oh(n).
\]
Strong connectivity gives the matching lower bound $m\ge n$. For every threshold,
\[
 \Phi_D\le\frac{CnD^{-p}}{1-Q^{-p}}.
\]
Substitution into the profile theorem yields, apart from the additive control term,
\[
 nDL^{3/2}\ell^{1/2}+\Oh(n^2D^{-p}L^2\ell).
\]
Equating these terms gives $D^{p+1}=\Theta(n\sqrt{L\ell})$. Rounding up produces the explicit profile bound
\begin{equation}\label{eq:explicitprofiletail}
 \Oh\!\left(n^{1+1/(p+1)}
 L^{3/2+1/(2(p+1))}\ell^{1/2+1/(2(p+1))}\right).
\end{equation}
For fixed $p$, the displayed bound's positive power above $n$ absorbs the additive $nL^2\ell$ term. The logarithms at $p=5$ are $L^{19/12}\ell^{7/12}$, as reported in Table~\ref{tab:bounds}.

For the pruned-density branch, $D\le n$ and $m=\Theta(n)$ imply $H_D=\Theta(n)$. Our pruned-density branch's cost is
\[
 \Ot(n+n^2D^{-p}+D^2+n^{6/5}).
\]
Balance $n^2D^{-p}$ and $D^2$ at $D=\lceil n^{2/(p+2)}\rceil$. The result is $\Ot(n^{4/(p+2)}+n^{6/5})$. For $1<p<4/3$, the exponent $4/(p+2)$ lies between $6/5$ and $4/3$ and is smaller than the profile exponent. For $4/3\le p\le4$, the density exponent $6/5$ is no larger than either alternative. For $p>4$, the profile exponent $1+1/(p+1)$ becomes smaller than $6/5$. This proves the three regimes in Corollary~\ref{cor:polytail}, including the endpoints without claiming a strict logarithmic improvement at an exponent tie.

The same calculation explains the improvement in the identifier range. A maximum-degree-only argument would infer $\Delta=\Oh(n^{1/p})$ and would need $p\ge5/3$ to reach the sparse $n^{6/5}$ guarantee through $\Delta\le\Ot(n^{3/5})$. The pruned argument reaches that guarantee already at $p=4/3$, because our analysis pays for the number of high receivers rather than their largest individual degree.

\emph{Reading the three regimes.} The density construction has an $n^{6/5}$ contribution after logarithms are suppressed. Better tails first reduce the pruning and exact-coverage terms to that level; still faster decay then makes the separate profile construction preferable. For $p=2$, the profile exponent is $4/3$, while density transfer gives $6/5$. For $p=5$, density transfer still gives $6/5$, whereas the profile exponent is $7/6$. The crossover values follow by equating $4/(p+2)$ and $1+1/(p+1)$, respectively, with $6/5$. Our portfolio compares these complete upper bounds automatically. The flat middle regime describes the cost of our density construction, rather than a lower bound on the underlying gossiping problem.

\subsection{Localized tails}
A global power law can be unnecessarily restrictive at small degrees. The next results permit arbitrary behavior below a threshold. The constants and tail parameters remain analytical comparators rather than protocol inputs.

\begin{corollary}[Localized profile bound]\label{cor:localizedprofile}
Suppose $h_k\le H/k^p$ for integers $k\ge k_0$, where $1\le k_0\le n$, $H\ge0$, and $p>0$ is fixed. Put
\[
 z=(H\sqrt{L\ell})^{1/(p+1)}.
\]
The profile portfolio completes gossiping in
\begin{equation}\label{eq:localizedprofile}
 \Oh\!\left(n\max\{k_0,z\}L^{3/2}\ell^{1/2}+nL^2\ell\right).
\end{equation}
\end{corollary}
\begin{proof}
Choose the threshold
\[
D=\min\{n,\max\{k_0,\lceil z\rceil\}\}.
\]
If $z\le n$, then $D\ge k_0$ and
\[
 \Phi_D\le\frac{HD^{-p}}{1-Q^{-p}}.
\]
Since $D\ge z$, the pruning contribution $nHD^{-p}L^2\ell$ is at most a constant times $nDL^{3/2}\ell^{1/2}$. The ceiling is absorbed by $k_0\ge1$. If $z>n$, use $D=n$ and $h_n=0$; the resulting $n^2L^{3/2}\ell^{1/2}$ cost is bounded by the right-hand side. The adaptive theorem removes $H,p,k_0$ as inputs.
\end{proof}

\begin{corollary}[Localized sparse transfer]\label{cor:localizeddensity}
Under the same localized tail condition, if $m=\Theta(n)$, the pruned-density pool gives
\begin{equation}\label{eq:localizeddensity}
 \Ot\!\left(n^{6/5}+\max\{k_0^2,(nH)^{2/(p+2)}\}\right).
\end{equation}
\end{corollary}
\begin{proof}
Set $z=(nH)^{1/(p+2)}$. Choose $D$ by rounding $z$ up, raising the threshold to at least $k_0$, and capping the threshold at $n$. When $z\le n$, the tail bound gives $n\rho_D=\Ot(nHD^{-p})$. Since $D\ge z$, this is at most $\Ot(D^2)$; $D^2=\Oh(\max\{k_0^2,z^2\})$. Also $E_D\le D^2$ and $H_D=\Theta(n)$, so Theorem~\ref{thm:density} proves~\eqref{eq:localizeddensity}. If $z>n$, $D=n$ has $\rho_n=0$ and $D^2\le z^2$, yielding the same conclusion. The general branch may further improve this bound. No assumption on degrees below $k_0$ was used.
\end{proof}

The localized profile bound pays linearly in $k_0$, whereas exact residual coverage pays quadratically. Their full expressions must therefore be compared. The shared portfolio makes this comparison algorithmic.

\subsection{Residual-edge refinement}
The maximum residual in-degree is not the only quantity improved by pruning. The same implementation can also exploit the number of arcs whose senders remain active, without discovering those arcs individually. For an arc $(x,y)$, inclusion in the residual count depends on whether $x$ remains active. The receiver $y$ may be dormant: delivery there can already establish quasi-dissemination. Thus the sum below counts active senders' arcs into all physical receivers. The padding by $n$ preserves the physical control budget even when few data arcs remain.

\begin{corollary}[Residual-edge-sensitive transfer]\label{cor:residualedges}
Fix a common threshold $D$. Let $F_0(D)$ be the output of the deterministic $\Prune(D)$ execution from the original state, let $A_D=V\setminus F_0(D)$, and define
\[
 \mu_D=\max\!\left\{n,\sum_v d_{A_D}(v)\right\}.
\]
Using only $n,N,D$, the fixed-threshold pruned-density protocol takes
\begin{equation}\label{eq:residualedges}
 \Ot\!\left(n(1+\rho_D)+E_D+
 [\mu_D(n+\sqrt{\mu_DD})]^{3/5}\right)
\end{equation}
rounds. The quantity $\mu_D$ is not an input.
\end{corollary}
\begin{proof}
Pruning depends on $n,N,D$ and the pruning procedure's private state, not on the edge guess used by condensing. Reset every trial to its own original rumors and invocation clocks. The deterministic pruning prefix is identical in all those trials and returns the same $F_0(D)$. The pruning prefix's cost is still bounded by Lemma~\ref{lem:charge}.

Lemma~\ref{lem:active} requires an upper bound on active arcs, padded to at least $n$; our condensing procedure does not require an upper bound on all physical edges. Therefore the first guess $M\ge\mu_D$ is valid. That guess is at most $2\mu_D$. The valid trial's condensing and complete replay have the bound in~\eqref{eq:residualedges}. Earlier trials cannot falsely pass the common all-rumors verification. Repeated pruning and the logarithmic number of guesses are absorbed by $\Ot$. No node needs to compute the sum defining $\mu_D$.
\end{proof}

Since $\mu_D\le m$, this refines the same fixed-$D$ bound. Control still pays for the full $n$-vertex physical graph. The dyadic pool inherits the refined guarantee for each threshold the pool actually contains. We do not assert that rounding an arbitrary threshold preserves $\mu_D$: a higher threshold can enlarge the residual edge count by more than a constant factor, even though the active sets are nested along the shared greedy trace. The all-threshold adaptation in Theorem~\ref{thm:density} uses the monotone original-profile bound, for which that rounding argument is valid.

The obstruction to rounding is concrete even on a complete directed graph. Let $b$ be the first amplification-$Q$ maximum, so $b\ge(1-1/Q)(n-1)$. Threshold $b+1$ stops before broadcasting and leaves all $n(n-1)$ arcs active. Threshold $b$ secures at least $b$ sources in its first broadcast, leaving at most $n-b\le n/Q+1$ active sources. Subsequent broadcasts can only reduce that set. Thus its padded residual edge count is at most $\Oh(n^2/Q+n)$, whereas the adjacent threshold leaves $\Theta(n^2)$. Since $Q$ grows with $n$, adjacent thresholds need not preserve the residual-edge budget within a constant factor.

\subsection{What the certificates do not imply}
\begin{proposition}[A count certificate is not designated delivery]\label{prop:scope}
The numerical inequalities in Lemma~\ref{lem:probe} alone do not imply that every neighboring packet, or a prescribed neighboring packet, is received. The all-or-threshold conclusion of Lemma~\ref{lem:union} alone does not imply delivery of a prescribed packet once the receiver has reached the threshold.
\end{proposition}
\begin{proof}
Take $d\ge Q$ active in-neighbors and suppose all but one are heard. Then $b=d-1\ge(1-1/Q)d$, so the count certificate permits the outcome. Put a designated packet only at the unheard neighbor. The designated packet's delivery does not follow from the numerical inequalities. This is a logical counterexample to the inference, not a lower bound on the actual schedule.

For the stronger interface, suppose the neighboring union contains $\pi$, while the receiver finishes with $r$ other designated packets and still lacks $\pi$. The alternative $|W_v|\ge r$ holds, so the disjunctive certificate also permits this outcome. A processed-layer control rule rules out that alternative in an unpaid frame. An unprocessed frontier has no such implication for the desired crossing, so exact coverage is retained there. The proposition does not exclude other protocols that force the missing packet by additional mechanisms.
\end{proof}

\subsection{A linear-time exponent is unavoidable in the sparse class}
\begin{proposition}[Distance lower bound within concentration]\label{prop:distance}
There are strongly connected graphs with $m=\Theta(n)$ satisfying~\eqref{eq:concentration}, with maximum in-degree $n-1$, on which every gossiping protocol takes $\Omega(n)$ rounds.
\end{proposition}
\begin{proof}
Start with the directed cycle $1\to2\to\cdots\to n\to1$ and add every missing arc into vertex $1$. The edge count is $2n-2$. All vertices other than $1$ have in-degree one, and vertex $1$ has in-degree $n-1$. With $\dbar=2-2/n$, fixed constants $C_0=3,C_1=1$ satisfy the concentration condition for all sufficiently large $n$.

The distance from $1$ to $n$ remains $n-1$: every added arc points back into $1$ and cannot shorten that route. Information traverses at most one directed edge per round, so the original rumor of $1$ cannot reach $n$ sooner. This proves the lower bound, regardless of collisions, identifiers, or message size. Our distance argument establishes the optimal linear polynomial exponent on this sparse family, without claiming matching logarithms or an $\Omega(m)$ lower bound for every density.
\end{proof}

\subsection{The remaining unrestricted extension problem}
The frontier cost and pruning charge identify a concrete target for further improvement. One would need either a cheaper designated frontier mechanism, a stronger way to collect many high-receiver inventories, or an amortization that charges frontier work across phases without requiring every current target to cross separately. Approximate neighbor counts alone cannot justify any of these steps.

The present proof also makes the frame cost explicit. Data invocations must fit inside the unpaid intervals used by the amortization. Increasing the identifier range does not merely add a logarithm to a round-robin argument: the larger range changes which invocations can fit and which packet guarantee can be inferred before the next frozen eligibility test. The packet certificate repairs the union guarantee, and pruning limits the residual exact work, but neither fact establishes the missing worst-case profile inequality. Thus the unrestricted extension is stated as open throughout the manuscript.

\section{Conclusion}
Neighborhood clearing secures an elected receiver's entire closed in-neighborhood within one election-scale budget. Charging each relevant neighborhood class once gives $\Oh(mL^2\ell)$ rounds under degree concentration for every polynomial identifier range, and $\Oh(nL^2\ell)$ on the distributed-cover and complete-core families. The complementary adaptive-accuracy procedure supports survivor-sensitive completion and reception covers. Path covers, frozen-active-set condensing, regularity, and degree tails provide further guarantees, with each bound evaluated on its own execution and complete-history replay used throughout.

The results extend the regular-graph identifier range, preserve the density benchmark throughout the stated range of $N$, and give nearly edge-linear gossiping under degree concentration. Preserving the density bound for every polynomial identifier range and attaining nearly edge-linear time on all sparse graphs remain open. The explicit frontier cost and pruning charges specify the steps that still require a stronger argument.

\appendices
\section{Protocol Realization and Detailed Cost Accounting}\label{app:protocol}
This appendix details the reception-score controller and logarithmic costs. Section~\ref{sec:clearing} gives the separate clearing controller. All branches use the physical control network and genuine rumor values; each data procedure specifies its senders.

\subsection{A common state at every decision boundary}
Each vertex stores its private inventory, common ledger $F$, and radio-action record. Only completed designated broadcasts update $F$. Thus $a=n-|F|$ is common without knowing all active identifiers; membership of its own identifier in $F$ determines a vertex's activity.

At a pruning boundary, all vertices know the padded elapsed budget $P$ and cap $e$. An iteration proceeds as follows.
\begin{enumerate}
\item Freeze the current active sender set, set $\kappa=\min(a,2e)$ and $q=\lceil nQ/\kappa\rceil$, and execute the corresponding reception schedule. Count distinct immediate sender identifiers. Save the resulting score and the complete actual rumor inventory.
\item Search all frozen keys $(n-b_v)N+\operatorname{id}(v)$. The search does not filter initiators by an analytical threshold. The key search's Boolean tests announce the same maximum score and winner to every vertex.
\item Broadcast only the winner's saved inventory for ledger credit. Actual packets may carry additional current knowledge, but crediting the saved set makes the update identical at every vertex.
\item At the completed broadcast boundary, update $F$, increase $P$ by the common padded iteration budget, and replace $e$ by $\min(e,b+1)$. If $a=0$, the first pass is complete. Otherwise compare $P$ with $C(a,e)$.
\end{enumerate}
Freezing snapshots before the search prevents its traffic from changing ledger credit and preserves common threshold prefixes. Transmissions may append current knowledge. Every cap update takes a minimum, so a bound installed by one winner persists through later winners.

If the stopping inequality holds, the outer loop ends permanently. Our protocol evaluates the three explicit completion upper bounds and runs a minimizing branch. The balanced branch computes $h$ from this handoff population and executes ordinary fixed-threshold $\Prune(h)$; it does not invoke the outer stopping rule recursively. This keeps the proof dependencies acyclic: completion uses only the previously proved reception, collection, charging, and path-selection lemmas.

The same record continues during completion. At the common ending boundary, each vertex repeats its recorded transmit/listen actions with all retained knowledge. Replay does not recompute scores, gates, eligibility, or stopping decisions. The second pass therefore has exactly the successful-reception opportunities used in Lemma~\ref{lem:replay}, although its packets contain more information.

\subsection{Clearing boundaries and implementation}
A clearing phase freezes scores and snapshots, searches all keys, and broadcasts the winner's snapshot, supplying the disjoint source charge. Later repetitions refresh the sender set and the same receiver's saved inventory. Only that snapshot receives credit. A secured receiver remains eligible; genuine rumors justify every source removal.

The horizon $t=\lceil\log_4 n\rceil$ requires no remaining-degree test. Repetitions with zero degree remain sound, and a fully secured state permits idling. Padding completes the common phase length $U$ before $P,e$ or the stopping rule changes. At a below-$D$ maximum, integrality gives $b+1\le D$. For $D=1$, a zero maximum forces $F=V$; for $D=n$, fixed-threshold pruning proceeds directly to completion.

More generally, any fixed integer amplification $q_0\ge2$ gives $d_{i+1}\le\lfloor d_i/q_0\rfloor$. With $t(q_0)=\lceil\log_{q_0}n\rceil$, clearing costs
\[
 \Oh\bigl(t(q_0)[nq_0\log N+R]\bigr)=\Oh(B).
\]
Higher amplification reduces repetitions but increases probe cost; the election already costs $\Theta(B)$. Amplification four suffices uniformly. Replay includes every probe, control action, and padding slot, without recomputing decisions.

\subsection{Empty tests, shared charges, and threshold rounding}
An empty frozen selection ends after the selection's first $R$-round OR: there is no winner, key search, or inventory broadcast. A successful selection costs at most $B$. Fresh flags prevent knowledge learned during a test from creating new initiators. This shorter empty boundary is needed for balanced completion; condensing may retain its own fixed padding.

Let $S_p,S_c$ count successful second-stage pruning and collection broadcasts. Their newly secured origins are disjoint subsets of the same $a$ entry-active sources. Thus
\[
 hS_p+(k/2)S_c\le a,\qquad S_p+S_c\le\lfloor a/h\rfloor.
\]
There are $q=\lceil\log k\rceil$ final empty tests and $q+1$ path passes, together with one final pruning probe and key search. The total is
\[
 \Oh\!\left(B(1+a/h)+(q+1)S_k+qR\right).
\]
Since $k<a\le n$, $qR\le B$ and $(q+1)S_k=\Oh(k^2\Lambda\log(k+1))$. No charge resets after a doubling iteration.

Put $g=\Lambda\log(a+1)$ and $x=(Ba/g)^{1/3}$. Since $g=\Oh(Ba)$, $h=\lceil x\rceil=\Theta(x)$ and the first power of two at least $4h$ satisfies $4h\le k<8h$. These choices use only $a,n,N$.

If $k\ge a$, direct path selection costs $\Oh(a^2g)=\Oh(x^2g)=\Oh(\mathcal J(a))$. Otherwise the stopping probe certifies $d_A<2h$, so closed contenders have size at most $2h<k$. If at most $k$ sources remain, direct selection suffices; if none remain, no schedule is needed. Otherwise the doubling loop applies. Every path pass freezes the current active set, every collection freezes the current witnesses, and later active neighborhoods are subsets of their entry neighborhoods. The data graph need not be connected.

Finally $Ba/x=x^2g$, and rounding preserves
\[
 Ba/h+k^2g=\Oh((Ba)^{2/3}g^{1/3}).
\]
For polynomial identifiers, substituting $B=\Theta(nL^2\ell)$ and $\Lambda=\Theta(L)$ gives
\[
 \mathcal J(a)=\Theta\!\left((na)^{2/3}L^{5/3}\ell^{2/3}
                    [\log(a+1)]^{1/3}\right).
\]
Since $1\le a\le n$, $\log(a+1)=\Oh(L)$, so this is at most $\Oh((na)^{2/3}L^2\ell^{2/3})$. The latter is a uniform upper bound; the preceding expression retains the sharper dependence on $a$. Both concern this completion construction, not a comparison with every general-purpose gossiping protocol.

\subsection{Which round-robin obligations have been removed}
Path completion avoids the exact frontier step of~\cite{wu}: path delivery and secured holders establish quasi-dissemination. The separate condensing branch retains exact coverage at an unprocessed frontier. Packet-union certification addresses congestion and collection, whose delivery requirements differ from prescribed frontier delivery.

Table~\ref{tab:interfaces} distinguishes the guarantees: sender counts certify interference, packet thresholds support collection, and ordered delivery assembles path inventories. Each global argument must convert its particular certificate into a secured source or useful holder; these delivery guarantees are not interchangeable.

\begin{table}[t]
\centering\small
\caption{Delivery guarantees and their roles in our protocol. Each invocation uses its stated frozen state and contender bound.}\label{tab:interfaces}
\renewcommand{\arraystretch}{1.10}
\begin{tabular}{@{}>{\raggedright\arraybackslash}p{0.27\columnwidth}>{\raggedright\arraybackslash}p{0.34\columnwidth}>{\raggedright\arraybackslash}p{0.27\columnwidth}@{}}
\toprule
Schedule & Delivery guarantee & Role\\
\midrule
Reception probe (Lemma~\ref{lem:probe}) & Certified count of distinct active senders heard. & Pruning; residual interference bound.\\
Packet-union probe (Lemma~\ref{lem:union}) & All neighboring entry packets, or a packet threshold. & Collection in processed classes.\\
Path selector (Lemma~\ref{lem:path}) & Entry inventories along covered carrier paths. & Completion; path-cover gains.\\
Exact one-hop schedule & Each active in-neighbor's entry packets. & Delivery across an unprocessed frontier.\\
\bottomrule
\end{tabular}
\end{table}

\subsection{Stopping rules, residual witnesses, and complete replay}
The clearing, reception-score, and inventory-score branches have different reference executions. The clearing branch compares only complete clearing phases and uses $\chi_D$; the survivor and cover quantities below keep their original traces. Their completion budgets are nonincreasing on their respective traces. Reception-score iterations and complete clearing phases each have constant cost $U$ on their own trace, yielding $T\le4\min_j\max\{jU,C_j\}+2U$. For $T_*=2\min_j(jU+C_j)$, this gives $T\le2T_*+2U$. The comparison uses prescribed upper budgets on one trace, not the fastest possible radio protocol.

The inventory branch also pays for path passes: $P_j\le3P_{j-1}+2U$ gives $2\mathcal T\le8\min_j\max\{P_j,C_j\}+4U$. This constant covers phase jumps. Both comparisons include $j=0$, empty terminal states, and integer rounding.

A direct residual cover at prefix $j$ constrains only $A_j$ minus the cover's exceptions. Later $A_i\subseteq A_j$, so the same receiver weights continue to cover every remaining nonexceptional source. A dormant receiver remains eligible. Successive stages may use unrelated weights and exceptions, provided each stage counts every omitted unsecured source. The incurred prefix index is always paid; an early handoff is justified by the prefix comparison rather than by executing later reference states.

For residual carrier-path covers, the entry carrier set stays fixed throughout geometric phases. Population completion instead permits the active set to shrink between frozen passes. On handoff the carrier rule can change because the cover argument has ended. Record all actions from the original rumors, including the pruning prefix. Replaying only a completion suffix is invalid when a holder learned a target after its own original rumor was disseminated.

\subsection{Scope of the accounting}
The cover criterion follows directly by requiring $n(j+\tau)$ and $(nu)^{2/3}$ to be at most $\Ot((mn)^{3/5})$, which gives $u=\Ot(m^{9/10}n^{-1/10})$. At linear density this allows $\Ot(n^{4/5})$ exceptions. In the complete-core family, $t_2\le3$ instead gives $\Oh(B)$ through the clearing branch. The general cover and quantile guarantees retain their own reception-score trace; this comparison does not transfer residual witnesses between branches.

Each parameter trial is accepted only after verification in its own simulation. The frozen OR of $|K_v|<n$ uses genuine source values, excluding packet headers and other simulations' inventories. A false answer therefore certifies full gossiping, whether or not the edge guess was valid or the graph met a favorable structural condition. The trial's charged cost includes pruning, collection, serialized channels, and replay. Thus the portfolio compares complete running-time bounds while retaining each proof's ledger, carrier, and witness assumptions.

\section*{Acknowledgment}
ChatGPT was used to improve the clarity, organization, and presentation of the manuscript; assist in simplifying the exposition of selected mathematical arguments; examine the logical consistency of proofs; and help refine the formulation of parameter bounds and clarify relevant boundary cases. The research problem, core ideas, algorithmic designs, and principal scientific contributions were independently developed by the author.

\end{document}